\documentclass[12pt,draftclsnofoot,onecolumn]{IEEEtran}

\newcommand{\diag}{\mathop{\rm diag}}

\makeatletter
\newcommand{\vast}{\bBigg@{4}}
\newcommand{\vastt}{\bBigg@{6}}
\makeatother
\usepackage[final]{graphicx}
\usepackage{subcaption}
\usepackage{color}
\usepackage{multirow}
\usepackage{graphicx}
\usepackage{epstopdf}
\usepackage[cmex10]{amsmath}
\usepackage{amssymb}

\usepackage{multirow}
\usepackage{algpseudocode}
\usepackage{amsthm}
\usepackage{cite}
\theoremstyle{definition}
\theoremstyle{remark}

\newtheorem{lem}{Lemma}
\newtheorem*{thm}{Theorem}
\newtheorem{cor}{Corollary}
\newtheorem*{prop}{Proposition}

\makeatletter
\g@addto@macro\th@remark{\thm@headpunct{\normalfont:}}
\makeatother
\makeatletter
\newcommand{\distas}[1]{\mathbin{\overset{#1}{\kern\z@\sim}}}%
\newsavebox{\mybox}\newsavebox{\mysim}
\newcommand{\distras}[1]{%
  \savebox{\mybox}{\hbox{\kern3pt$\scriptstyle#1$\kern3pt}}%
  \savebox{\mysim}{\hbox{$\sim$}}%
  \mathbin{\overset{#1}{\kern\z@\resizebox{\wd\mybox}{\ht\mysim}{$\sim$}}}%
	}

\let\emptyset\varnothing
\allowdisplaybreaks

\begin{document}

\title{Simultaneous Spectrum Sensing and Data Reception for Cognitive Spatial Multiplexing Distributed Systems}

\author{Nikolaos~I.~Miridakis, Theodoros~A.~Tsiftsis,~\IEEEmembership{Senior Member,~IEEE},\\ George C. Alexandropoulos,~\IEEEmembership{Senior Member,~IEEE},\\ and M\'erouane Debbah,~\IEEEmembership{Fellow,~IEEE}
\thanks{N. I. Miridakis is with the Department of Computer Engineering, Piraeus University of Applied Sciences, 12244 Aegaleo, Greece (e-mail: nikozm@unipi.gr).}
\thanks{T. A. Tsiftsis is with the Department of Electrical Engineering, Technological Educational Institute of Central Greece, 35100 Lamia, Greece and with the School of Engineering, Nazarbayev University, 010000 Astana, Kazakhstan (e-mails: tsiftsis@teiste.gr, theodoros.tsiftsis@nu.edu.kz).}
\thanks{G. C. Alexandropoulos and M. Debbah are with the Mathematical and Algorithmic Sciences Lab, France Research Center, Huawei Technologies France, 92100 Boulogne-Billancourt, France (e-mails: \{george.alexandropoulos,merouane.debbah\}@huawei.com).}
}

\markboth{}%
{Simultaneous Spectrum Sensing and Data Reception for Distributed Cognitive Spatial Multiplexing Distributed Systems}

\maketitle

\begin{abstract}
A multi-user cognitive (secondary) radio system is considered, where the spatial multiplexing mode of operation is implemented amongst the nodes, under the presence of multiple primary transmissions. The secondary receiver carries out minimum mean-squared error (MMSE) detection to effectively decode the secondary data streams, while it performs spectrum sensing at the remaining signal to capture the presence of primary activity or not. New analytical closed-form expressions regarding some important system measures are obtained, namely, the outage and detection probabilities; the transmission power of the secondary nodes; the probability of unexpected interference at the primary nodes; {\color{black}and the detection efficiency with the aid of the area under the receive operating characteristics curve}. The realistic scenarios of channel fading time variation and channel estimation errors are encountered for the derived results. Finally, the enclosed numerical results verify the accuracy of the proposed framework, while some useful engineering insights are also revealed, such as the key role of the detection accuracy to the overall performance and the impact of transmission power from the secondary nodes to the primary system. 
\end{abstract}

\begin{IEEEkeywords}
Cognitive radio, detection probability, imperfect channel estimation, minimum mean-squared error (MMSE), outage probability, spatial multiplexing, spectrum sensing.
\end{IEEEkeywords}

\IEEEpeerreviewmaketitle

\section{Introduction}
\IEEEPARstart{C}{ognitive} radio (CR) has emerged as one of the most promising technologies to resolve the issue of spectrum scarcity, caused by the escalating growth in wireless data traffic of next-generation networks \cite{j:deepsensing}. One of the principal requirements of CR is the effectiveness of spectrum sharing performed by secondary (unlicensed) nodes, which is expected to intelligently mitigate any harmful interference caused to the primary (licensed) network nodes. This requirement is directly related to the accuracy of spectrum sensing techniques, reflecting the reliable detection of primary transmission(s).

On the other hand, placing multiple antennas on each cognitive node represents a fruitful option since the system capacity in terms of data rate can be greatly enhanced. Spatial multiplexing represents one of the most prominent techniques used for multiple input-multiple output (MIMO) transmission systems \cite{j:gesbert}. For computational savings at the receiver side, there has been a prime interest in the class of linear detectors, such as zero-forcing (ZF) and minimum mean-squared error (MMSE). It is widely known that MMSE outperforms ZF, especially in low-to-medium signal-to-noise (SNR) regions, at the cost of a slightly higher computational burden, since the noise variance is required in this case. In addition, when MIMO technology is combined with distributed antenna systems (DAS), the so-called distributed-MIMO (D-MIMO) transmission is emerged. The success behind D-MIMO relies on the multiplexing gains, which are produced by the classical MIMO transmission, and the diversity gains, which are manifested from the use of DAS \cite{j:das}.

Due to the complementary benefits of CR and D-MIMO, the cognitive (D-)MIMO systems are of paramount research interest nowadays, e.g., see \cite{j:Sboui,j:Gupta,j:Zhongyuan,j:Sharawi} and references therein.

\subsection{Related Work and Motivation}
The performance of spectrum sensing, i.e., the accuracy of the detection method used by the cognitive system plays a key role to the performance of both the primary and secondary network. It acts as an important tool for finding idle spectrum instances (the so-called \emph{spectrum holes} \cite{j:Haykin}) to efficiently deliver cognitive data, while protecting the communication quality of the primary service at the same time. Several spectrum sensing approaches have been proposed so far to preserve transparency of CR networks, which can be categorized into two main types; quiet \cite{j:Peh} and active \cite{j:Song}.

The quiet spectrum sensing type is the conventional approach in which each potential cognitive transmitter first senses the spectrum for a fixed time-sensing duration and then transmits its data in the remaining time, when it senses the channel as idle. The main problem of this approach is the capacity reduction in terms of cognitive data transmission within a given frame duration. Moreover, the detection accuracy is questionable by adopting the quiet type, since the sensing-time duration is rather limited (i.e., only a fraction of the entire frame duration) and, hence, the required number of sensing samples is constrained.

In order to overcome this problem, the more sophisticated active sensing type has been proposed. The idea behind this approach relies on the improvement of the former shortages produced by quiet sensing. In particular, a simultaneous spectrum sensing and data transmission technique was proposed in \cite{c:Stotas}, where the receiver first cancels the secondary data using interference cancellation and then senses the remaining signal for the presence or absence of a primary activity. However, the scenario of a single transmitter-receiver pair for the cognitive system was considered in \cite{c:Stotas} with the presence of only one primary node, a rather infeasible condition for practical applications. Other active sensing techniques for multi-user cognitive systems were proposed in \cite{j:Song} and \cite{j:Moghimi}. In both studies, it was assumed that some secondary nodes transmit while others perform spectrum sensing. In the case of a primary signal detection, the latter nodes inform the former ones about the primary activity to stop their transmissions. Nevertheless, several problems arise by following these methods; more spectrum resources are required because of the signaling overhead caused by the informing process, whereas extra power resources are consumed from the sensing nodes during spectrum sensing and because of transmitting their sensing reports. 

More recently, authors in \cite{j:Tsakalaki} and \cite{j:Jihaeng} proposed a spatial isolation technique on the antennas of each cognitive node in a sense that some antennas are devoted for spectrum sensing while others for data transmission. The main drawback of this approach is the large amount of self-interference produced during spectrum sensing, which can not always be sufficiently canceled. Hardware constraints and/or impairments represent an immediate obstacle, whereas an appropriate physical distance between the sensing and transmitting antennas should be maintained (i.e., in the order of $20-40$cm \cite{j:Duarte,c:Bharadia}), which is not always feasible or preferable for simple small-sized equipment.

In addition, the concept of simultaneous data reception and spectrum sensing for single-antenna nodes was studied in \cite{j:jeong2011channel,c:Karunakaran2014LTE}, while for multiple-antenna nodes in \cite{c:Karunakaran2014}. However, these works used the central limit theorem to approximate the total received signal as a Gaussian input (invoking the constraint of sufficiently large amount of received samples), whereas they provided only semi-analytical and/or simulation results with respect to the system performance.

Capitalizing on the aforementioned observations, in this paper, a new simultaneous (active) spectrum sensing and data transmission approach for CR networks is presented. The spectrum sensing is performed at the secondary receiver upon the overall signal reception from multiple secondary transmitters. The spatial multiplexing mode of operation is adopted, for the first time, where all the potential secondary transmitters send their data streams simultaneously in a given frame duration. Thus, the self-interference problem is tackled, since all antennas at the receiver are used first for signal detection/decoding for the secondary data and then for spectrum sensing in the same frame duration. The receiver utilizes the linear MMSE approach to efficiently detect the secondary streams. Since the noise variance is, in principle, a requisite for the MMSE detection/decoding, the optimum energy detector (ED) can be used for the following spectrum sensing process (which also requires the knowledge of the noise variance). However, since the spectrum sensing is implemented at the receiver, it is possible that a primary activity in the vicinity of one or more secondary transmitters may not be sensed by the receiver, mainly due to the different link distances and/or independent signal propagation losses. To avoid the latter \emph{hidden terminal} problem, a distributed power allocation scheme is implemented by each secondary transmitter, upon signal transmission. Based on this scheme, each secondary transmitter appropriately adjusts its power in order not to cause any harmful interference to the potentially active primary node(s), preserving transparency of the secondary activity.  

Overall, the main benefits of this work are twofold: (a) an efficient tradeoff between sensing time and data transmission time and its relevant computation is no longer an issue; and (b) the self-interference problem is effectively mitigated, since the simultaneous transmission and spectrum sensing are implemented by different (i.e., sufficiently separated in terms of transmission wavelength) nodes.

\subsection{Contributions}
The contributions of this work are summarized as follows:

\begin{itemize}
	\item A new mode-of-operation and protocol design for cognitive networks is presented and analytically described. The novelty of this scheme relies on the fact that it uses the spatial multiplexing transmission scheme, where multiple single-antenna secondary nodes may send their streams simultaneously to a multiple-antenna secondary receiver, under the presence of multiple primary nodes. Independent and non-identically distributed (i.n.i.d.) statistics are considered, suitable for practical networking setups (i.e., different link-distances amongst the primary and secondary nodes). To this end, the considered secondary system forms a (virtual) D-MIMO infrastructure. The receiver simultaneously performs signal detection and spectrum sensing in the same frame duration. Further, a distributed power allocation scheme is applied on the involved secondary transmitters. 
	\item New analytical closed-form expressions are derived for some important system measures when all signals undergo Rayleigh channel fading, namely, the outage and detection probabilities, the transmission power for each secondary node and the probability of unexpected interference at the primary nodes.
	\item \color{black}{As it is explicitly indicated in the upcoming analysis, the accuracy of the detection scheme plays a key role to the system performance. Thereby, we further investigate the detection performance with the aid of the receive operating characteristics (ROC) curves, and a solid performance measure, the area under the ROC curve (AUC). A new exact closed-form expression of AUC for the considered system is also obtained.} 
	\item For the above derivations, the channel aging effect and channel estimation errors are both considered. In other words, the analysis incorporates outdated channel state information (CSI) and/or imperfect CSI for i.n.i.d. Rayleigh fading channels. The results are simplified for the scenario of independent and identically distributed (i.i.d.) statistics.
\end{itemize}

\subsection{Organization of the paper}
The rest of this paper is organized as follows. This Section continues with some notational definitions for the most important mathematical symbols used in the subsequent analysis. In Section \ref{System Model}, the considered system model and the proposed mode of operation are described in detail. Key statistical derivations regarding the received signal-to-interference-plus-noise ratio (SINR) are obtained in Section \ref{Statistics of SINR}. In Section \ref{Performance Metrics}, the considered system is thoroughly analyzed, whereas the aforementioned performance measures are obtained in closed form. Further, important insights regarding the transmission power used by the secondary nodes are presented in Section \ref{Impact of the Transmission Power Used by the Secondary Network}. In Section \ref{Numerical Results}, the proposed framework is validated and cross-compared with simulation results, while some useful engineering insights are revealed. Finally, Section \ref{Conclusion} concludes the paper.  

\emph{Notation}: Vectors and matrices are represented by lowercase bold typeface and uppercase bold typeface letters, respectively. Also, $\mathbf{X}^{-1}$ is the inverse of $\mathbf{X}$ and $\mathbf{x}_{i}$ denotes the $i$th coefficient of $\mathbf{x}$. A diagonal matrix with entries $x_{1},\cdots,x_{n}$ is defined as $\diag\{x_{i}\}^{n}_{i=1}$. The superscripts $(\cdot)^{T}$ and $(\cdot)^{\mathcal{H}}$ denote transposition and Hermitian transposition, respectively, $\|\cdot\|$ corresponds to the vector Euclidean norm, while $|\cdot|$ represents absolute (scalar) value. In addition, $\mathbf{I}_{v}$ stands for the $v\times v$ identity matrix, $\mathbb{E}[\cdot]$ is the expectation operator, $\overset{\text{d}}=$ represents equality in probability distributions and $\text{Pr}[\cdot]$ returns probability. Also, $f_{X}(\cdot)$ and $F_{X}(\cdot)$ represent probability density function (PDF) and cumulative distribution function (CDF) of the random variable (RV) $X$, respectively. Complex-valued Gaussian RVs with mean $\mu$ and variance $\sigma^{2}$, while chi-squared RVs with $v$ degrees-of-freedom are denoted, respectively, as $\mathcal{CN}(\mu,\sigma^{2})$ and $\mathcal{X}^{2}_{2v}$. Furthermore, $\Gamma(a)\triangleq (a-1)!$ (with $a\in \mathbb{N}^{+}$) denotes the Gamma function \cite[Eq. (8.310.1)]{tables}, $\Gamma(\cdot.,\cdot)$ is the upper incomplete Gamma function \cite[Eq. (8.350.2)]{tables}, while $(\cdot)_{p}$ is the Pochhammer symbol with $p \in \mathbb{N}$ \cite[p. xliii]{tables}. Further, $J_{0}(\cdot)$ represents the zeroth-order Bessel function of the first kind \cite[Eq. (8.441.1)]{tables}, ${}_1F_{1}(\cdot,\cdot;\cdot)$ denotes the Kummer's confluent hypergeometric function \cite[Eq. (9.210.1)]{tables}, ${}_2F_{1}(\cdot,\cdot,\cdot;\cdot)$ denotes the Gauss hypergeometric function \cite[Eq. (9.100)]{tables}, and $Q_{\nu}(\cdot,\cdot)$ is the generalized $\nu$th order Marcum-$Q$ function \cite{b:marcum}.

\section{System Model}
\label{System Model}
Consider a cognitive (secondary) communication system, which is consisted of $m_{c}$ single-antenna cognitive transmitters and a receiver equipped with $N\geq m_{c}$ antennas\footnote{It follows from the subsequent analysis that the considered system is equivalent to the case when a single cognitive transmitter is used equipped with $m_{c}$ antennas.} operating under the presence of $m_{p}$ single-antenna primary nodes. Notice that although $N\geq m_{c}$ is a necessary condition in order to capture the available degrees-of-freedom during the detection of the streams from the cognitive transmitting nodes, it holds that $N\lessgtr (m_{p}+m_{c})$. Moreover, i.n.i.d. Rayleigh flat fading channels are assumed, reflecting non-equal distances among the involved nodes with respect to the receiver, an appropriate condition for practical applications.

The spatial multiplexing mode of operation is implemented in the secondary system, where $m_{c}$ independent data streams are simultaneously transmitted by the corresponding secondary nodes. A suboptimal yet quite efficient detection scheme is adopted, the so-called linear MMSE, which is performed at the secondary receiver.

Letting $M\triangleq m_{p}+m_{c}$, the received signal at the $n$th sample time-instance reads as
\begin{align}
\mathbf{y}[n]=\hat{\mathbf{H}}[n] \mathbf{s}[n]+\mathbf{w}[n],
\label{eq1}
\end{align} 
where $\mathbf{y}[n] \in \mathbb{C}^{N\times 1}$, $\hat{\mathbf{H}}[n] \in \mathbb{C}^{N\times M}$, $\mathbf{s}[n] \in \mathbb{C}^{M\times 1}$ and $\mathbf{w}[n] \in \mathbb{C}^{N\times 1}$ denote the received signal, the estimated channel matrix, the transmitted signal and the additive white Gaussian noise (AWGN), respectively. It holds that $\mathbf{w}\overset{\text{d}}=\mathcal{CN}(\mathbf{0},N_{0}\mathbf{I}_{N})$ with $N_{0}$ denoting the AWGN variance and $\mathbf{s}=[s_{1},\ldots,s_{m_{p}},s_{1},\ldots,s_{m_{c}}]^{T}$ with $\mathbb{E}[\mathbf{s}\mathbf{s}^{\mathcal{H}}]=\mathbf{I}_{M}$. In addition, $\hat{\mathbf{H}}=[\hat{\mathbf{h}}_{1},\ldots,\hat{\mathbf{h}}_{m_{p}},\hat{\mathbf{h}}_{1},\ldots,\hat{\mathbf{h}}_{m_{c}}]$, whereas $\hat{\mathbf{h}}_{i}\overset{\text{d}}=\mathcal{CN}(\mathbf{0},\beta_{i}\mathbf{I}_{N})$, for $1\leq i\leq M$, with $\beta_{i}\triangleq p_{i}/(d_{i}^{\omega_{i}})$, where $p_{i}$, $d_{i}$, and $\omega_{i}$ correspond to the signal power, normalized estimated distance (with a reference distance equal to $1$km) from the receiver and path-loss exponent of the $i$th transmitter, respectively.

\subsection{Protocol Description}
\begin{figure}
\begin{subfigure}{.5\textwidth}
  \centering
  \includegraphics[width=.4\linewidth]{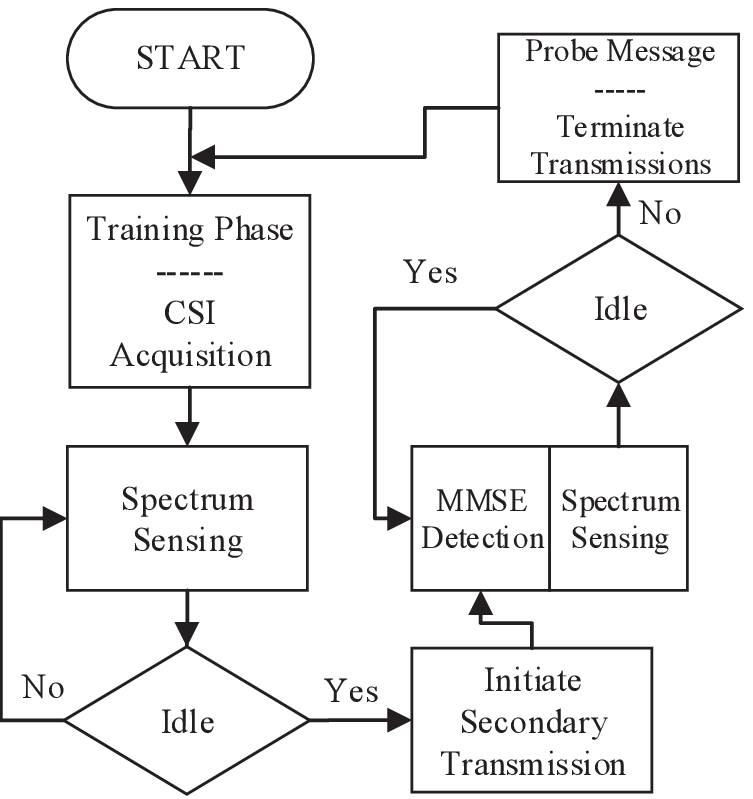}
	\caption{}
  \label{fig1}
\end{subfigure}%
\begin{subfigure}{.5\textwidth}
  \centering
  \includegraphics[width=1.0\linewidth]{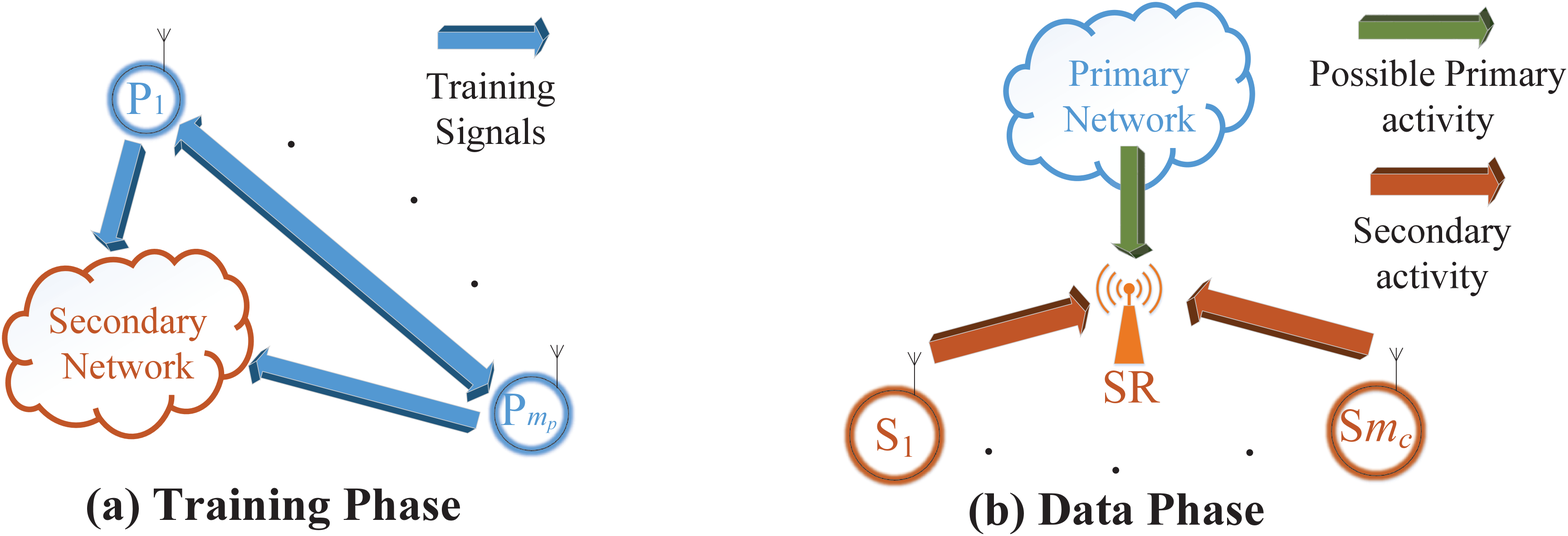}
	\caption{}
  \label{fig01}
\end{subfigure}
\caption{a) Flowchart of the proposed mode of operation at the secondary receiver; b) The considered system configuration, where S$_i$, P$_j$ and SR stand for the $i$th secondary transmitter (with $1\leq i\leq m_{c}$), $j$th primary transmitter (with $1\leq j\leq m_{p}$) and secondary receiver, respectively.}
\label{fig:fig}
\end{figure}
The mode of operation for the considered cognitive system is constituted by three main phases; namely, the \emph{training}, \emph{data transmission} and \emph{spectrum sensing} phases, which are periodically alternating.

In the training phase, all the involved nodes (i.e., primary and secondary transmitters) broadcast certain (orthogonal) pilot signals. The secondary receiver monitors the available spectrum resources in order to acquire the instantaneous channel gains from all the existing nodes (both primary and secondary). Meanwhile, all the secondary transmitters also monitor the channel in order to acquire channel gains between the primary nodes and themselves. This occurs in order to appropriately modify their power, which will use for potential transmission in the subsequent data phase. It is assumed that the channel remains constant during this phase. However, its status may change in subsequent time instances.  

Afterwards, the system enters the data phase, where the secondary nodes stay inactive for one symbol-time duration. During this time period, the secondary receiver senses the spectrum so as to capture the presence of a primary communication activity or not. In the former case, no transmission activity is performed by the secondary transmitters (lack of triggering from the secondary receiver in this case is interpreted as a busy spectrum notification to all the transmitters). This procedure is repeated in every subsequent symbol-time duration, until the receiver senses the spectrum idle. In the latter case, the receiver broadcasts a certain probe message in order to initiate the secondary transmission(s). Hence, in the next symbol-time instance, all active secondary nodes may simultaneously transmit their data streams. Upon the overall signal reception, MMSE detection is performed at the secondary receiver and all data streams are decoded concurrently.

After the removal of all secondary signals from the received signal, the spectrum sensing phase is implemented (within the same symbol-time instance), where the receiver monitors the remaining signal for the presence of a potential primary activity. If the remaining signal is sensed idle (i.e., only the presence of noise), the same procedure keeps on (i.e., data transmission-spectrum sensing), until the next training phase. If at least one primary signal is detected at the remaining signal, then the receiver immediately broadcasts another certain message in order to coarsely finalize all the secondary transmissions. An appropriate ceiling on the transmission power of the receiver is utilized in order not to cause unexpected co-channel interference to the primary communication(s). Similarly, all the active secondary transmitters use a relevant ceiling for their transmissions due to the same reason (explicit details on this ceiling are provided into the next section). The basic lines of reasoning of the proposed scheme are sketched in Figs. \ref{fig1} and \ref{fig01}. 

It is noteworthy that the motivation behind the proposed system configuration relies on certain conditions and/or limitations, which are viable in various realistic networking implementations. More specifically, conventional CR services based on television white spaces may use the traditional quiet spectrum sensing (without the training phase requirement). This is because the number of TV channels is limited (approximately $50\mathtt{\sim}70$, each with a bandwidth $6-8$MHz, within a total spectrum range between $54-862$MHz \cite{c:Cordeiro}). In this spectrum range, spectrum sensing time is indeed acceptable, whereas most IEEE $802.22$ equipments are for indoor installation and, hence, their power consumption is not an actual problem \cite{j:Ling}. Nonetheless, more sophisticated CR services, such as the IEEE $1900.4$ standard \cite{j:Buljore}, are designed to use spectrum resources from multiple radio-access-technology (RAT) heterogeneous primary networking systems, e.g., cellular systems. Consequently, spectrum range for these systems is emphatically increased (e.g., $450$MHz-$3$GHz), while the spectrum sensing time and the corresponding energy cost are extremely increased in this case, thus, becoming non-efficient. To this end, training-based signaling, which is, in principle, utilized for primary cellular configurations can be used from the cognitive/secondary system to perform spectrum sensing and/or acquire important statistics, such as channel gains and transmission powers of primary nodes/users. Besides, long-term evolution (LTE) has initiated a CR-based operation quite recently \cite{3gpp}, under the concept of licensed-assisted access/licensed-shared access (LAA/LSA), which operates with the aid of training (pilot) signaling \cite{c:Miguelez,DOCOMO,lopes2015}.      

\subsection{Training Phase: Channel Estimation}
During the training phase, $M$ orthogonal pilot sequences (i.e., unique spatial signal signatures) of length $M$ symbols are assigned to the primary and cognitive nodes.\footnote{In various network setups, primary users periodically transmit training signals intended for primary receivers to assist them in channel estimation and/or synchronization \cite[\S 12.3.1]{b:jayaweera2014signal}. Building on this feature, secondary nodes can overhear these transmissions to capture their own estimates amongst the primary nodes and themselves. The first step is to enable the training process for the secondary nodes along with the primary ones. Doing so, the secondary receiver is able to acquire CSI statistics from both networks.} Then, the received pilot signal can be expressed as
\begin{align}
\mathbf{Y}_{\text{tr}}[n]=\mathbf{H}_{\text{tr}}[n] \mathbf{\Psi}+\mathbf{W}_{\text{tr}}[n],
\label{trainingsignal}
\end{align}  
where $\mathbf{Y}_{\text{tr}}[n] \in \mathbb{C}^{N\times M}$, $\mathbf{H}_{\text{tr}}[n] \in \mathbb{C}^{N\times M}$, $\mathbf{\Psi} \in \mathbb{C}^{M\times M}$ and $\mathbf{W}_{\text{tr}}[n] \in \mathbb{C}^{N\times M}$ denote the received signal, the channel matrix, the transmitted pilot signals and AWGN, respectively, all representing the training phase. Also, the pilot signals are normalized satisfying $\mathbb{E}[\mathbf{\Psi}\mathbf{\Psi}^{\mathcal{H}}]=\mathbf{I}_{M}$.

The MMSE channel estimate of $\mathbf{h}_{i}[n]$, $1\leq i\leq M$, is given by \cite[Eq. (10)]{j:Truong_massive_Mimo} $\hat{\mathbf{h}}_{i}[n]=\beta_{i}\left(N_{0}+\sum^{M}_{j=1}\beta_{j}\right)^{-1}\mathbf{I}_{N}\left(\sum^{M}_{j=1}\mathbf{h}_{j}[n]+\mathbf{w}_{\text{tr}}[n]\right)$, where $\mathbf{w}_{\text{tr}}[n]$ is the AWGN at the $i$th channel during the training phase. It is noteworthy that with MMSE channel estimation, the channel estimate and the channel estimation error remain uncorrelated (i.e., due to the orthogonality principle \cite{b:stevenkayesttheory}). In particular, we have that
\begin{align}
\hat{\mathbf{h}}_{i}[n]=\mathbf{h}_{i}[n] +\tilde{\mathbf{h}}_{i}[n],\ \ 1\leq i\leq M,
\label{channel decomposition}
\end{align}
where $\mathbf{h}_{i}\overset{\text{d}}=\mathcal{CN}(\mathbf{0},(\beta_{i}-\hat{\beta}_{i})\mathbf{I}_{N})$ is the true channel fading of the $i$th transmitter and $\tilde{\mathbf{h}}_{i}\overset{\text{d}}=\mathcal{CN}(\mathbf{0},\hat{\beta}_{i}\mathbf{I}_{N})$ denotes its corresponding estimation error with $\hat{\beta}_{i}\triangleq \beta^{2}_{i}/(\sum^{M}_{j=1}\beta_{j}+N_{0})$ \cite[Eq. (12)]{j:Truong_massive_Mimo}.

Except the channel estimation errors, the channel aging effect occurs in several practical network setups. This is mainly because of the rapid channel variations during consecutive sample time-instances, due to, e.g., user mobility and/or severe fast fading conditions. The popular autoregressive (Jakes) model of a certain order \cite{jakes}, based on Gauss-Markov block fading channel, can accurately capture the latter effect. More specifically, it holds that
\begin{align}
\hat{\mathbf{h}}_{i}[n]=\alpha^{M}\hat{\mathbf{h}}_{i}[n-M] +\sum^{M-1}_{m=0}\alpha^{m}\mathbf{e}_{i}[n-m],
\label{channel aging}
\end{align}  
where $\alpha\triangleq J_{0}(2\pi f_{D}T_{s})$ with $f_{D}$ and $T_{s}$ denoting the maximum Doppler shift and the symbol sampling period, respectively. Moreover, $\mathbf{e}'_{i}\triangleq \sum^{M-1}_{m=0}\alpha^{m}\mathbf{e}_{i}[n-m]$ stands for the stationary Gaussian channel error vector due to the time variation of the channel, which is uncorrelated with $\mathbf{h}_{i}[n-M]$, while $\mathbf{e}'_{i}\overset{\text{d}}=\mathcal{CN}(\mathbf{0},(1-\alpha^{2 M})\beta_{i}\mathbf{I}_{N})$. For the sake of mathematical simplicity and without loss of generality, we assume that the channel remains unchanged over the time period of training phase, while it may change during the subsequent data transmission phase. Thus, adopting the autoregressive model of order one, (\ref{channel aging}) simplifies to 
\begin{align}
\hat{\mathbf{h}}_{i}[n]=\alpha\hat{\mathbf{h}}_{i}[n-1] +\mathbf{e}_{i}[n].
\label{channel aging1}
\end{align} 
Substituting (\ref{channel decomposition}) into (\ref{channel aging1}), we have that\footnote{In what follows, the time-instance index $n$ is dropped for ease of presentation, since all the involved random vectors are mutually independent.}
\begin{align}
\hat{\mathbf{h}}_{i}=\alpha \mathbf{h}_{i}+\alpha \tilde{\mathbf{h}}_{i}+\mathbf{e}_{i}\triangleq \mathbf{g}_{i}+\boldsymbol{\epsilon}_{i},
\label{channeljoint}
\end{align} 
where $\mathbf{g}_{i}\overset{\text{d}}=\mathcal{CN}(\mathbf{0},(\beta_{i}-\hat{\beta}_{i})\alpha^{2}\mathbf{I}_{N})$ and $\boldsymbol{\epsilon}_{i}\overset{\text{d}}=\mathcal{CN}(\mathbf{0},\alpha^{2}\hat{\beta}_{i}+(1-\alpha^{2})\beta_{i})\mathbf{I}_{N})$.

It should be noted that the latter model in (\ref{channeljoint}) combines both the channel aging effect and the channel estimation error. Hence, by defining $\mathbf{G}\triangleq [\mathbf{g}_{1},\ldots,\mathbf{g}_{m_{p}},\mathbf{g}_{1},\ldots,\mathbf{g}_{m_{c}}]$ and $\mathbf{E}\triangleq [\boldsymbol{\epsilon}_{1},\ldots,\boldsymbol{\epsilon}_{m_{p}},\boldsymbol{\epsilon}_{1},\ldots,\boldsymbol{\epsilon}_{m_{c}}]$, (\ref{eq1}) can be reformulated as
\begin{align}
\mathbf{y}=\mathbf{G}\mathbf{s}+\mathbf{E}\mathbf{s}+\mathbf{w}.
\label{eq1new}
\end{align} 

\subsection{Data Transmission Phase: Signal Detection}
Benefiting from the training phase whereby estimating the channel gains of all the signals, the cognitive receiver proceeds with the detection/decoding of the simultaneously transmitted streams from the $m_{c}$ cognitive nodes. The mean-squared error (MSE) of the $i$th received stream ($1\leq i\leq m_{c}$) is formed as
\begin{align}
\text{MSE}_{i}=\mathbb{E}\left[\left|s_{i}-\boldsymbol{\phi}_{i}^{\mathcal{H}}\mathbf{y}\right|^{2}\right],
\label{mse}
\end{align} 
where $\boldsymbol{\phi}_{i}$ is the optimal weight vector.

\begin{cor}
The optimal weight vector, which minimizes MSE of the $i$th received stream is given by
\begin{align}
\boldsymbol{\phi}_{i}=\sqrt{\beta_{i}}\left(\mathbf{C}\diag\{\beta_{j}\}^{M}_{j=1}\mathbf{C}^{\mathcal{H}}+N_{0}\mathbf{I}_{N}\right)^{-1}\mathbf{c}_{i},
\label{phi}
\end{align} 
where $\mathbf{C} \in \mathbb{C}^{N\times M}$ and $\mathbf{C}\overset{\text{d}}=\mathcal{CN}(\mathbf{0},\mathbf{I}_{N})$, while $\mathbf{c}_{i}$ is its $i$th column vector.
\end{cor}

\begin{proof}
The proof of (\ref{phi}) is relegated in Appendix \ref{appa}.
\end{proof}

At the receiver, $\boldsymbol{\phi}^{\mathcal{H}}_{i}\mathbf{y}$ is utilized for the detection of the $i$th transmitted stream, yielding
\begin{align}
\nonumber
z_{i}&=\boldsymbol{\phi}^{\mathcal{H}}_{i}\mathbf{y}=\boldsymbol{\phi}^{\mathcal{H}}_{i}\mathbf{g}_{i}s_{i}+\sum_{j\neq i}\boldsymbol{\phi}^{\mathcal{H}}_{i}\mathbf{g}_{j}s_{j}+\boldsymbol{\phi}^{\mathcal{H}}_{i}\mathbf{E}\mathbf{s}+\boldsymbol{\phi}^{\mathcal{H}}_{i}\mathbf{w}\\
&=\underbrace{\left(\mathbf{A}^{-1}\mathbf{g}_{i}\right)^{\mathcal{H}}\mathbf{g}_{i}}_{\triangleq \mathcal{P}_{i}}s_{i}+\underbrace{\left(\mathbf{A}^{-1}\boldsymbol{\epsilon}_{i}\right)^{\mathcal{H}}\mathbf{g}_{i}s_{i}+\sum_{j\neq i}\boldsymbol{\phi}^{\mathcal{H}}_{i}\mathbf{g}_{j}s_{j}+\boldsymbol{\phi}^{\mathcal{H}}_{i}\mathbf{E}\mathbf{s}+\boldsymbol{\phi}^{\mathcal{H}}_{i}\mathbf{w}}_{\triangleq \mathcal{R}_{i}},
\label{zsignal}
\end{align}
where $\mathbf{A}\triangleq \mathbf{C}\diag\{\beta_{j}\}^{M}_{j=1}\mathbf{C}^{\mathcal{H}}+N_{0}\mathbf{I}_{N}$.

\subsection{Spectrum Sensing}
ED is the optimum detection method, since channel gains, signal, and noise variances are all known (or estimated) \cite{j:Tandra}. In addition, the use of multiple antennas at the secondary receiver can overcome the estimation uncertainty and improve the performance of spectrum sensing, by exploiting many available observations in the spatial domain \cite{ietbayesian}. Let the remaining signal, after decoding the $m_{c}$ secondary signals (thus, after removing their impact from the remaining signal), be defined as $\mathbf{r}$. Then, (\ref{eq1new}) becomes
\begin{align}
\mathbf{r}=\mathbf{G}_{p}\mathbf{s}_{p}+\mathbf{E}_{p}\mathbf{s}_{p}+\mathbf{w}=\mathbf{C}_{p}\diag\{\sqrt{\beta_{i}}\}^{m_{p}}_{i=1}\mathbf{s}_{p}+\mathbf{w},
\label{remsignal}
\end{align}
where $\mathbf{r} \in \mathbb{C}^{N\times 1}$, $\mathbf{G}_{p} \in \mathbb{C}^{N\times m_{p}}$, $\mathbf{E}_{p} \in \mathbb{C}^{N\times m_{p}}$, $\mathbf{C}_{p} \in \mathbb{C}^{N\times m_{p}}$ and $\mathbf{s}_{p} \in \mathbb{C}^{m_{p}\times 1}$ denote the remaining received signal, the true channel matrix, the estimation error matrix, the equivalent (joint) channel matrix and the transmitted signal from the primary nodes, respectively. Also, $\mathbf{C}_{p}\overset{\text{d}}=\mathcal{CN}(\mathbf{0},\mathbf{I}_{N})$.

{\color{black}In practice, perfect removal of the $m_{c}$ secondary signals (after decoding) may not always be the case due to, e.g., hardware constraints and/or impairments at the secondary receiver. Hence, this process may cause residual noise onto the remaining signal prior to spectrum sensing. In this case, (\ref{remsignal}) becomes
\begin{align}
\mathbf{r}=\mathbf{C}_{p}\diag\{\sqrt{\beta_{i}}\}^{m_{p}}_{i=1}\mathbf{s}_{p}+\mathbf{w}',
\label{remsignal11}
\end{align}
where $\mathbf{w}'\triangleq \mathbf{w}+\mathbf{w_{\epsilon}}$ with $\mathbf{w_{\epsilon}}$ being the additive post-noise after the aforementioned imperfect cancellation/removal. Assuming that $\mathbf{w_{\epsilon}}$ is a zero-mean Gaussian distributed vector \cite{j:UmarSheikh2014,j:Tandra}, we can model the post-noise as $\mathbf{w_{\epsilon}} \in \mathbb{R}^{N \times 1}$ while $\mathbf{w_{\epsilon}}\overset{\text{d}}=\mathcal{N}(\mathbf{0},\sigma^{2}_{\epsilon}\mathbf{I}_{N})$, where $\sigma^{2}_{\epsilon}$ denotes the level of impact due to imperfect cancellation of the secondary signals. Typically, the value of $\sigma^{2}_{\epsilon}$ can be captured by the secondary receiver via measurements during operation \cite{c:Kalamkar2013} and/or is predetermined from the system manufacturer. With known $\sigma^{2}_{\epsilon}$, the total noise $\mathbf{w}'$ is modeled as $\mathbf{w}'\overset{\text{d}}=\mathcal{CN}(\mathbf{0},\hat{N_{0}}\mathbf{I}_{N})$ with $\hat{N_{0}}\triangleq N_{0}+\sigma^{2}_{\epsilon}$.\footnote{In what follows, for ease of presentation and without loss of generality, $\hat{N_{0}}=N_{0}$ is assumed (which implies that $\mathbf{w}'=\mathbf{w}$). On the other hand, when the variance of post-noise $\mathbf{w_{\epsilon}}$ is not available, exact closed formulations for the detection and false alarm probabilities are not feasible; yet, current ones (presented in the following section) can serve as a performance benchmark or upper performance bounds.}}

Thereby, the binary hypothesis test is formed as
\begin{align}
T_{\text{ED}}\triangleq \sum^{L-1}_{l=0}\left\|\mathbf{r}(l)\right\|^{2}\overset{\mathcal{H}_{1}}{\underset{\mathcal{H}_{0}}{\lessgtr}} \lambda,
\label{ted}
\end{align}
where $L$ and $\lambda$ denote the number of samples for the received signal and the energy threshold, respectively. Moreover, the two hypotheses $\mathcal{H}_{0}$ and $\mathcal{H}_{1}$ correspond to the cases of no primary signal transmission and \emph{at least} one primary signal transmission, respectively. They are explicitly defined by the structure of the received signal's covariance matrix as  
\begin{align}
\begin{array}{l l l}     
    \mathcal{H}_{0}: &\mathbb{E}[\mathbf{r}\mathbf{r}^{\mathcal{H}}]=N_{0}\mathbf{I}_{N},\text{ no signal is present} & \\
    \mathcal{H}_{1}: &\mathbb{E}[\mathbf{r}\mathbf{r}^{\mathcal{H}}]=\text{ any positive semi-definite matrix.} &
\end{array}
\label{hypo}
\end{align}

\section{Statistics of SINR}
\label{Statistics of SINR}
We commence by defining the SINR for each stream with its corresponding CDF with respect to the cognitive communication performance, followed by the false alarm and detection probabilities with respect to the spectrum sensing performance. Then, the unconditional outage probability of the considered system is formulated in a closed form. Then, other important system measures are also obtained in closed form, namely, AUC, the transmission power of the secondary nodes, and the probability of unexpected co-channel interference at the primary nodes. 

Notice from (\ref{eq1new}) that $(\mathbf{g}_{i}+\boldsymbol{\epsilon}_{i})\overset{\text{d}}=\sqrt{\beta_{i}}\mathbf{c}_{i}$ and, thus, it is straightforward to show that $\boldsymbol{\epsilon}_{i}\overset{\text{d}}=\mathcal{CN}(\mathbf{0},\alpha^{2}\hat{\beta}_{i}+(1-\alpha^{2})\beta_{i}\mathbf{I}_{N})=(\sqrt{\alpha^{2}\hat{\beta}_{i}+(1-\alpha^{2})\beta_{i}})\mathbf{c}_{i}$. Hence, it follows that
\begin{align}
\mathbf{A}^{-1}\boldsymbol{\epsilon}_{i}=\mathbf{A}^{-1}\left(\sqrt{\frac{\alpha^{2}\hat{\beta}_{i}+(1-\alpha^{2})\beta_{i}}{\beta_{i}}}\right)\boldsymbol{\phi}_{i},
\label{Pidistr}
\end{align}
while based on (\ref{Pidistr}), we have from (\ref{zsignal}) that
\begin{align}
\mathcal{P}_{i}\overset{\text{d}}=\left((\beta_{i}-\hat{\beta}_{i})\alpha^{2}\right)\left(\mathbf{A}^{-1}\mathbf{c}_{i}\right)^{\mathcal{H}}\mathbf{c}_{i}.
\label{Pidistr1}
\end{align}
Using the above methodology, it also holds from (\ref{zsignal}) that
\begin{align}
\mathcal{R}_{i}=\boldsymbol{\phi}^{\mathcal{H}}_{i}\mathbf{C}\mathbf{I}'^{(i)}_{M}\mathbf{s}+\boldsymbol{\phi}^{\mathcal{H}}_{i}\mathbf{w},
\label{Ridistr}
\end{align}
where $\mathbf{I}'^{(i)}_{M}$ is a special diagonal matrix, which is formed as
\begin{align}
\mathbf{I}'^{(i)}_{M}\triangleq \left\{
\begin{array}{cr}     
    \diag\left\{\sqrt{\beta_{j}}\right\}^{M}_{j=1,j\neq i}\\
    & \\
    \sqrt{\alpha^{2}\hat{\beta}_{i}+(1-\alpha^{2})\beta_{i}},\text{ for the $i$th position}.
\end{array}\right.
\label{Ispecial}
\end{align}
Thereby, since $\mathbb{E}[\mathcal{R}_{i}\mathcal{R}^{\mathcal{H}}_{i}]=\boldsymbol{\phi}^{\mathcal{H}}_{i}(\mathbf{C}(\mathbf{I}'^{(i)}_{M})^{2}\mathbf{C}^{\mathcal{H}}+N_{0}\mathbf{I}_{N})\boldsymbol{\phi}_{i}$, the SINR of the $i$th transmitted stream reads as
\begin{align}
\nonumber
&\text{SINR}_{i}=\frac{\mathcal{P}^{2}_{i}}{\mathbb{E}[\mathcal{R}_{i}\mathcal{R}^{\mathcal{H}}_{i}]}=\frac{\left((\beta_{i}-\hat{\beta}_{i})\alpha^{2}\left(\mathbf{A}^{-1}\mathbf{c}_{i}\right)^{\mathcal{H}}\mathbf{c}_{i}\right)^{2}}{\boldsymbol{\phi}^{\mathcal{H}}_{i}(\mathbf{C}(\mathbf{I}'^{(i)}_{M})^{2}\mathbf{C}^{\mathcal{H}}+N_{0}\mathbf{I}_{N})\boldsymbol{\phi}_{i}}=\frac{\left((\beta_{i}-\hat{\beta}_{i})\alpha^{2}\left(\mathbf{A}^{-1}\mathbf{c}_{i}\right)^{\mathcal{H}}\mathbf{c}_{i}\right)^{2}}{\beta_{i}\left(\mathbf{A}^{-1}\mathbf{c}_{i}\right)^{\mathcal{H}}_{i}(\mathbf{C}(\mathbf{I}'^{(i)}_{M})^{2}\mathbf{C}^{\mathcal{H}}+N_{0}\mathbf{I}_{N})\left(\mathbf{A}^{-1}\mathbf{c}_{i}\right)}\\
&\approx \frac{\left((\beta_{i}-\hat{\beta}_{i})\alpha^{2}\left(\mathbf{A}^{-1}\mathbf{c}_{i}\right)^{\mathcal{H}}\mathbf{c}_{i}\right)^{2}}{\beta_{i}\left(\mathbf{A}^{-1}\mathbf{c}_{i}\right)^{\mathcal{H}}_{i}\mathbf{c}_{i}}=\left(\frac{\left((\beta_{i}-\hat{\beta}_{i})\alpha^{2}\right)^{2}}{\beta_{i}}\right)\mathbf{c}^{\mathcal{H}}_{i}\left(\mathbf{C}\diag\{\beta_{j}\}^{M}_{j=1}\mathbf{C}^{\mathcal{H}}+N_{0}\mathbf{I}_{N}\right)^{-1}\mathbf{c}_{i}.
\label{sinr}
\end{align} 
The approximation stage in the latter expression is formed by assuming that $(\mathbf{C}(\mathbf{I}'^{(i)}_{M})^{2}\mathbf{C}^{\mathcal{H}}+N_{0}\mathbf{I}_{N})\approx \mathbf{C}\diag\{\beta_{j}\}^{M}_{j=1}\mathbf{C}^{\mathcal{H}}+N_{0}\mathbf{I}_{N}=\mathbf{A}$. It becomes exact in the case when perfect CSI conditions occur.\footnote{As indicated from the numerical results provided in Section \ref{Numerical Results}, the approximation error remains negligible in moderate channel estimation error conditions.}

Based on Woodbury's identity \cite[Eq. (2.1.5)]{b:woodbury}, (\ref{sinr}) can alternatively be expressed as
\begin{align}
\text{SINR}_{i}=\left(\frac{\left((\beta_{i}-\hat{\beta}_{i})\alpha^{2}\right)^{2}}{\beta_{i}}\right)\frac{\Phi_{i}}{1+\Phi_{i}},
\label{sinrr}
\end{align} 
where $\Phi_{i}\triangleq \mathbf{c}^{\mathcal{H}}_{i}\left(\mathbf{K}\diag\{\beta_{j}\}^{M}_{\begin{subarray}{c}j=1\\j\neq i\end{subarray}}\mathbf{K}^{\mathcal{H}}+N_{0}\mathbf{I}_{N}\right)^{-1}\mathbf{c}_{i}$, while $\mathbf{K}\triangleq \sum_{j\neq i}\mathbf{c}_{j}$. The form of (\ref{sinrr}) is preferable than (\ref{sinr}) for further analysis, because $\mathbf{c}_{i}$ and $\mathbf{K}$ are statistically independent.

\begin{lem}
The CDF of SINR for the $i$th transmitted stream for a system with $M$ simultaneous transmitting nodes and $N$ receive antennas, while $1\leq i\leq m_{c}$, is presented in a closed form as
\begin{align}
F^{(N\times M)}_{\text{SINR}_{i}}(x)=F^{(N\times M)}_{\Phi_{i}}\left(\frac{x}{\frac{\left((\beta_{i}-\hat{\beta}_{i})\alpha^{2}\right)^{2}}{\beta_{i}}-x}\right),
\label{cdfsinr}
\end{align} 
where 
\begin{align}
\nonumber
F^{(N\times M)}_{\Phi_{i}}(y)=&1-\exp\left(-\frac{N_{0}}{\beta_{i}}y\right)\Bigg[\sum^{N}_{i=1}\frac{(\frac{N_{0}}{\beta_{i}}y)^{i-1}}{(i-1)!}-\sum^{N}_{i=\max\{1,N-M+1\}}\frac{\left(\frac{N_{0}}{\beta_{i}}y\right)^{i-1}}{(i-1)!}\\
&\times \frac{\displaystyle \sum^{M}_{\begin{subarray}{l}j=N-i+1\\j\neq i\end{subarray}}\:\:\sum_{1\leq n_{1}<\cdots<n_{j}\leq M}\textstyle \left(\frac{\beta_{n_{1}}}{\beta_{i}}\frac{\beta_{n_{2}}}{\beta_{i}}\cdots \frac{\beta_{n_{j}}}{\beta_{i}}\right)y^{j}}{\displaystyle \prod^{M}_{\begin{subarray}{c}n=1\\n\neq i\end{subarray}}\textstyle \left(1+\frac{\beta_{n}}{\beta_{i}}y\right)}\Bigg].
\label{phicdf}
\end{align} 
When $\beta_{1}=\beta_{2}=\cdots=\beta_{M}\triangleq \beta$, (\ref{phicdf}) reduces to
\begin{align}
&F^{(N\times M)}_{\Phi_{i}}(y)=1-\exp\left(-\frac{N_{0}}{\beta}y\right)\Bigg[\sum^{N}_{i=1}\frac{(\frac{N_{0}}{\beta}y)^{i-1}}{(i-1)!}-\sum^{N}_{i=\max\{1,N-M+2\}}\frac{\left(\frac{N_{0}}{\beta}y\right)^{i-1}\displaystyle \sum^{M-1}_{j=N-i+1}\textstyle \binom{M-1}{j}y^{j}}{(i-1)!(1+y)^{M-1}}\Bigg].
\label{phicdfiid}
\end{align}
\end{lem}

\begin{proof}
The proof is provided in Appendix \ref{appb}.
\end{proof}
The CDF in (\ref{phicdfiid}), implies identical channel fading conditions for all the nodes (i.e., equal distances with regards to the receiver), which is a rather infeasible scenario. Nonetheless, it can be used as a performance benchmark and/or a good approximation when $\beta_{1}\approx \beta_{2}\approx \cdots \approx \beta_{M}$.

\section{Performance Metrics}
\label{Performance Metrics}

\subsection{Detection Probability}
It suffices to show that in the case of $\mathcal{H}_{1}$ hypothesis, even if only the weakest signal is present, $T_{\text{ED}}>\lambda$ should hold. The latter condition can be modeled as
\begin{align}
\mathbf{r}_{\min}=\sqrt{\beta_{\min}}\mathbf{c}_{\min}s_{\min}+\mathbf{w},
\label{remsignalmin}
\end{align}
where $\mathbf{r}_{\min}$ represents the remaining received signal, when only the primary node with the weakest channel gain (at the secondary receiver) is active. The transmitted signal from the corresponding primary node is defined as\footnote{In general, the signal variance can be estimated by the sample variance for sufficiently large number of samples as $\sigma^{2}_{p}\approx (1/L)\sum^{L-1}_{l=0}\left|s_{\min}(l)\right|^{2}-((1/L)\sum^{L-1}_{l=0}s_{\min}(l))^{2}$. If the sample mean goes to zero, then $\sigma^{2}_{p}\approx (1/L)\sum^{L-1}_{l=0}\left|s_{\min}(l)\right|^{2}$.} $s_{\min}$ with $\mathbb{E}[s_{\min}s_{\min}^{\mathcal{H}}]=\sigma^{2}_{p}$. Also, $\sqrt{\beta_{\min}}\mathbf{c}_{\min}$ satisfies that $\beta_{\min}\left\|\mathbf{c}_{\min}\right\|^{2}=\min\{\beta_{\min}\left\|\mathbf{c}_{p,i}\right\|^{2}\}^{m_{p}}_{i=1}$, where $\mathbf{c}_{p,i}$ represents the $i$th column vector of $\mathbf{C}_{p}$. Notice that a Gaussian vector is isotropically distributed, i.e., it remains Gaussian distributed even if its norm is under some constraint \cite[Theorem 1.5.5]{b:multivariate}. Thus, $\sqrt{\beta_{\min}}\mathbf{c}_{\min}\overset{\text{d}}=\mathcal{CN}(\mathbf{0},\beta_{\min}\mathbf{I}_{N})$ and $\beta_{\min}\left\|\mathbf{c}_{\min}\right\|^{2}$ is the minimum of $m_{p}$ non-identical $\mathcal{\chi}^{2}_{2N}$ RVs.

\begin{lem}
The PDF of $\mathcal{Y}\triangleq \beta_{\min}\left\|\mathbf{c}_{\min}\right\|^{2}$ is presented in a closed-form as
\begin{align}
f_{\mathcal{Y}}(x)=\sum^{m_{p}}_{s=1}\sum^{N-1}_{\begin{subarray}{c}t_{1}=0\\t_{1}\neq t_{s}\end{subarray}}\cdots \sum^{N-1}_{\begin{subarray}{c}t_{m_{p}}=0\\t_{m_{p}}\neq t_{s}\end{subarray}}\frac{\beta^{-t_{1}}_{1}\cdots \beta^{-N}_{s}\cdots \beta^{-t_{m_{p}}}_{m_{p}}}{t_{1}!\cdots t_{m_{p}}!\Gamma(N)} x^{\sum^{m_{p}}_{\begin{subarray}{c}l=1\\l\neq s\end{subarray}}t_{l}+N-1}\exp\left(-\left(\displaystyle \sum^{m_{p}}_{t=1}\textstyle \frac{1}{\beta_{t}}\right)x\right).
\label{pdfmin}
\end{align}
\end{lem}

\begin{proof}
The CDF of $\mathcal{Y}$ stems as
\begin{align}
\text{Pr}[\mathcal{Y}<x]=1-\left(\text{Pr}[\beta_{1}\left\|\mathbf{c}_{1}\right\|^{2}>x]\cdots \text{Pr}[\beta_{m_{p}}\left\|\mathbf{c}_{m_{p}}\right\|^{2}>x]\right).
\end{align}
Using the standard complementary CDF of a $\mathcal{\chi}^{2}_{2N}$ RV into the previous expression yields
\begin{align}
F_{\mathcal{Y}}(x)=1-\prod^{m_{p}}_{t=1}\frac{\Gamma \left(N,\frac{x}{\beta_{t}}\right)}{\Gamma(N)}.
\label{cdfmin}
\end{align}
By differentiating (\ref{cdfmin}), it holds that
\begin{align}
f_{\mathcal{Y}}(x)=\sum^{m_{p}}_{s=1}\frac{x^{N-1}\exp\left(-\frac{x}{\beta_{s}}\right)}{\Gamma(N)\beta^{N}_{s}}\prod^{m_{p}}_{\begin{subarray}{c}t=1\\t\neq s\end{subarray}}\frac{\Gamma\left(N,\frac{x}{\beta_{t}}\right)}{\Gamma(N)}.
\label{ppdfmin}
\end{align}
Further, expanding $\Gamma(.,.)$ in terms of finite sum series according to \cite[Eq. (8.352.4)]{tables}, (\ref{pdfmin}) is obtained.
\end{proof}

The detection probability is defined as $P_{d}\triangleq \text{Pr}[T_{\text{ED}}|\mathcal{H}_{1}>\lambda]$. In the case of ED, it is given by \cite[Eq. (63)]{j:Wang}
\begin{align}
P_{d}(\lambda)=Q_{N L}\left(\sqrt{\frac{2 L \sigma^{2}_{p}\mathcal{Y}}{N_{0}}},\sqrt{\frac{\lambda}{N_{0}}}\right).
\label{pdcond}
\end{align}

\begin{cor}
The unconditional detection probability of the considered system with $N$ receive antennas and $m_{p}$ active primary nodes is presented in a closed form as
\begin{align}
P^{(N\times m_{p})}_{d}(\lambda)=\sum^{m_{p}}_{s=1}\sum^{N-1}_{\begin{subarray}{c}t_{1}=0\\t_{1}\neq t_{s}\end{subarray}}\cdots \sum^{N-1}_{\begin{subarray}{c}t_{m_{p}}=0\\t_{m_{p}}\neq t_{s}\end{subarray}}\frac{\beta^{-t_{1}}_{1}\cdots \beta^{-N}_{s}\cdots \beta^{-t_{m_{p}}}_{m_{p}}}{t_{1}!\cdots t_{m_{p}}!\Gamma(N)} \mathcal{F}\left(\sum^{m_{p}}_{\begin{subarray}{c}l=1\\l\neq s\end{subarray}}t_{l}+N,N L,\sqrt{\frac{2 L \sigma^{2}_{p}}{N_{0}}},\sqrt{\frac{\lambda}{N_{0}}},\sum^{m_{p}}_{t=1}\frac{1}{\beta_{t}}\right),
\label{pd}
\end{align}
where
\begin{align}
\nonumber
\mathcal{F}\left(k,m,a,b,p\right)\triangleq &\frac{\Gamma(k)\Gamma(m,\frac{b^{2}}{2})}{p^{k}\Gamma(m)}+\frac{a^{2}b^{2m}\Gamma(k)\exp\left(-\frac{b^{2}}{2}\right)}{m!p^{k}2^{m}(a^{2}+2p)} \sum^{k-1}_{l=0}\left(\frac{2p}{a^{2}+2p}\right)^{l}\\
&\times \underbrace{{}_1F_{1}\left(l+1,m+1;\frac{a^{2}b^{2}}{2a^{2}+4p}\right)}_{\mathcal{T}}.
\label{1f1}
\end{align}
\end{cor}

\begin{proof}
Based on (\ref{pdcond}), the unconditional detection probability is evaluated as
\begin{align}
P_{d}(\lambda)=\int^{\infty}_{0}Q_{N L}\left(\sqrt{\frac{2 L \sigma^{2}_{p}x}{N_{0}}},\sqrt{\frac{\lambda}{N_{0}}}\right)f_{\mathcal{Y}}(x)dx.
\label{pddd}
\end{align}
Plugging (\ref{pdfmin}) into (\ref{pddd}), integrals of the following form appear
\begin{align}
\int^{\infty}_{0}x^{k-1}\exp(-p x)Q_{m}(a\sqrt{x},b)dx, \ \ \{a,b,m,p,k\}\geq 0.
\label{int}
\end{align}
Fortunately, such integrals were analytically evaluated in \cite[Eq. (12)]{j:sofotasiosmarcum}. Thus, using the latter result into (\ref{pddd}) and after performing some straightforward manipulations, (\ref{pd}) arises.
\end{proof}

At this point, it should be stated that when the first two arguments of ${}_{1}F_{1}(.,.;.)$ are non-negative integers, this expression can be relaxed to finite sum series including simple elementary functions, according to \cite[Eq. (7.11.1.10)]{b:prudnikovvol3}. In fact, this is the case presented in (\ref{1f1}), returning only simple elementary functions, which reads as
\begin{align}
\mathcal{T}=\left\{
\begin{array}{l l}     
    \exp\left(\frac{a^{2}b^{2}}{2a^{2}+4p}\right)\displaystyle \sum^{l-m}_{k=0}\textstyle \frac{(m-l)_{k}\left(-\frac{a^{2}b^{2}}{2a^{2}+4p}\right)^{k}}{k!(m+1)_{k}}, & l\geq m\\
    & \\
    \frac{(m-1)!(-m)_{l+1}}{l!\left(\frac{a^{2}b^{2}}{2a^{2}+4p}\right)^{m}}\Bigg(\displaystyle \sum^{m-l-1}_{k=0}\textstyle \frac{(l-m+1)_{k}\left(\frac{a^{2}b^{2}}{2a^{2}+4p}\right)^{k}}{k!(1-m)_{k}} & \\
		-\exp\left(\frac{a^{2}b^{2}}{2a^{2}+4p}\right)\displaystyle \sum^{l}_{k=0}\textstyle \frac{(-l)_{k}\left(-\frac{a^{2}b^{2}}{2a^{2}+4p}\right)^{k}}{k!(1-m)_{k}}\Bigg), & l<m
\end{array}\right.
\label{T}
\end{align}

\subsection{False Alarm Probability and Threshold Design}
The scenario of a false alarm probability, namely, $P_{f}(\lambda)$, can be modeled by $P_{f}(\lambda)\triangleq \text{Pr}[T_{\text{ED}}|\mathcal{H}_{0}>\lambda]$. Under the $\mathcal{H}_{0}$ hypothesis, $T_{\text{ED}}$ is the sum of the square of $N L$ independent and identically distributed Gaussian RVs with zero mean and variance $N_{0}$, i.e, $T_{\text{ED}}\overset{\text{d}}=N_{0}\mathcal{\chi}^{2}_{2 N L}$. Hence, using the standard complementary CDF of a chi-square RV, it yields
\begin{align}
P_{f}(\lambda)=\frac{\Gamma\left(N L, \frac{\lambda}{2 N_{0}}\right)}{\Gamma(N L)}.
\label{pf}
\end{align} 

As it is obvious from (\ref{pf}), the false alarm probability is an \emph{offline} operation, i.e., it is independent from the instantaneous channel gain and the number of primary signals. Thus, a convenient yet effective strategy is to select the optimum energy threshold using (\ref{pf}). Doing so, it holds that
\begin{align}
\lambda^{\ast}=P^{-1}_{f}(\tau),
\label{thr}
\end{align} 
where $\lambda^{\ast}$ represents the optimum energy threshold for a predetermined target $\tau$ (on the false alarm probability), while $P^{-1}_{f}(.)$ denotes the inverse function of $P_{f}(.)$, which can be efficiently calculated by using well-known inverse algorithms, e.g., \cite{j:didonato}.

Afterwards, the \emph{online} detection probability can be directly computed by calculating $P^{(N\times m_{p})}_{d}(\lambda^{\ast})$, using (\ref{pd}).

\subsection{Outage Probability}
Based on the above key analytical results, we are now in a position to formulate the outage probability of the considered system. Outage probability of the $i$th stream ($1\leq i \leq m_{c}$), $P^{(i)}_{\text{out}}(\gamma_{\text{th}})$, is defined as the probability that the SINR of the $i$th stream falls below a certain threshold value $\gamma_{\text{th}}\triangleq 2^{\mathcal{R}}-1$, where $\mathcal{R}$ stands for a given data transmission rate in bps/Hz.

\begin{thm}
The outage probability of the $i$th stream ($1\leq i \leq m_{c}$) is presented in a closed form as
\begin{align}
\nonumber
P^{(i)}_{\text{out}}(\gamma_{\text{th}})=&\left(1-P_{f}(\lambda^{\ast})\right)P^{p}_{\mathcal{A}}\{\emptyset\}F^{(N\times m_{c})}_{SINR_{i}}(\gamma_{\text{th}})+\sum^{M}_{j=m_{c}+1}\left(1-P^{(N\times (j-m_{c}))}_{d}(\lambda^{\ast})\right)\\
&\times \prod^{j-m_{c}}_{d_{1}=1}P^{p}_{\mathcal{A}}\{d_{1}\}\prod^{M-j}_{d_{2}=1}\left(1-P^{p}_{\mathcal{A}}\{d_{2}\}\right)F^{(N\times j)}_{SINR_{i}}(\gamma_{\text{th}}),
\label{out}
\end{align} 
where $P^{p}_{\mathcal{A}}\{\Delta\}$ represents the probability that $\Delta$ primary nodes are active, while $P^{p}_{\mathcal{A}}\{\emptyset\}$ denotes that there is no active primary node (i.e., an empty set) at the given time-instance.
\end{thm}

\begin{proof}
The proof is given in Appendix \ref{appe}.
\end{proof}
It is noteworthy that $P^{p}_{\mathcal{A}}\{\cdot\}$ is directly related to the transmission arrival rate of each primary node. For instance, a typical model used thoroughly in wireless systems for the distribution of data traffic is the widely known Poisson process \cite{b:goldsmith}. In this case, $P^{p}_{\mathcal{A}}\{\cdot\}$ follows the inter-arrival exponential distribution modeled as $P^{p}_{\mathcal{A}}\{x\}=\exp(-v T_{s})$, where $v$ is the average transmission arrival rate. Nevertheless, the analysis and/or the efficient modeling of transmission arrival rates represents a research topic out of the scope of current work. 

{\color{black}\subsection{Area Under the ROC Curve}
The accuracy of ED plays a crucial role to the outage probability, which is reflected on the underlying detection and false alarm statistics. Due to this reason, we further investigate the performance of ED using a more solid measure, the so-called AUC. The main benefit of AUC is that it jointly evaluates the performance of both the detection and false alarm in the entire energy threshold region.

More specifically, the conditional AUC (on a given channel gain) is defined as \cite[Eq. (5)]{j:aparattu}
\begin{align}
\text{AUC}(\mathcal{Y})=-\int^{\infty}_{0}P_{d}(\lambda')\frac{\partial P_{f}(\lambda')}{\partial \lambda'}d\lambda',
\label{pdpf}
\end{align}
where $\lambda'$ stands for the normalized energy threshold $\lambda'\triangleq \lambda/N_{0}$. 

\begin{cor}
The conditional AUC of the considered ED scheme is presented in a closed form as
\begin{align}
\text{AUC}(\mathcal{Y})=1-\exp\left(-\frac{L\sigma^{2}_{p}\mathcal{Y}}{N_{0}}\right)\sum^{N L-1}_{l=0}\frac{(N L)_{l}}{l!2^{N L+l}} {}_1F_{1}\left(N L+l,N L;\frac{L\sigma^{2}_{p}\mathcal{Y}}{2 N_{0}}\right).
\label{condauc}
\end{align}
\end{cor}

\begin{proof}
The detailed proof is presented in Appendix \ref{appj}.
\end{proof}
Averaging (\ref{condauc}) over the PDF of $\mathcal{Y}$, the unconditional (average) AUC is presented as follows.
\begin{prop}
The unconditional AUC is given by
\begin{align}
\nonumber
\overline{\text{AUC}}=&1-\sum^{N L-1}_{l=0}\sum^{l}_{k=0}\sum^{m_{p}}_{s=1}\sum^{N-1}_{\begin{subarray}{c}t_{1}=0\\t_{1}\neq t_{s}\end{subarray}}\cdots \sum^{N-1}_{\begin{subarray}{c}t_{m_{p}}=0\\t_{m_{p}}\neq t_{s}\end{subarray}}\frac{(N L)_{l}\beta^{-t_{1}}_{1}\cdots \beta^{-N}_{s}\cdots \beta^{-t_{m_{p}}}_{m_{p}}\Gamma\left(\sum^{m_{p}}_{\begin{subarray}{c}l=1\\l\neq s\end{subarray}}t_{l}+N\right)}{l!2^{N L+l}t_{1}!\cdots t_{m_{p}}!\Gamma(N)\left(\sum^{m_{p}}_{t=1}\frac{1}{\beta_{t}}+\frac{L\sigma^{2}_{p}}{N_{0}}\right)^{\sum^{m_{p}}_{\begin{subarray}{c}l=1\\l\neq s\end{subarray}}t_{l}+N}}\\
&\times \frac{(-l)_{k}\left(N L-\sum^{m_{p}}_{\begin{subarray}{c}l=1\\l\neq s\end{subarray}}t_{l}\right)_{k}\left(\frac{L \sigma^{2}_{p}}{2N_{0}\left(\sum^{m_{p}}_{t=1}\frac{1}{\beta_{t}}+\frac{L\sigma^{2}_{p}}{N_{0}}\right)}\right)^{k}}{k!(N L)_{k}\left(1-\frac{L \sigma^{2}_{p}}{2N_{0}\left(\sum^{m_{p}}_{t=1}\frac{1}{\beta_{t}}+\frac{L\sigma^{2}_{p}}{N_{0}}\right)}\right)^{\sum^{m_{p}}_{\begin{subarray}{c}l=1\\l\neq s\end{subarray}}t_{l}+l}}.
\label{uncondauc}
\end{align}
\end{prop}

\begin{proof}
The proof is relegated in Appendix \ref{appk}.
\end{proof}
}

\section{Impact of the Transmission Power Used by the Secondary Network}
\label{Impact of the Transmission Power Used by the Secondary Network}

\subsection{Transmission Power of Secondary Nodes}
First, we define the transmission power of the receiver in the case of the aforementioned signaling process (c.f., Fig. \ref{fig1}). Recall that in the case when the receiver senses the spectrum busy (idle) by a primary transmission, upon an ongoing secondary communication, then it immediately informs the secondary nodes to terminate (initiate) their transmissions using a certain probe message. In order not to cause an additional co-channel interference to the potentially active primary node(s), the power used for this message is appropriately upper bounded. Particularly, it is defined as\footnote{Note that $p_{R}$ is the fixed power of the secondary receiver, whereas it is determined by the $Q_{R}$ statistic, which is computed during the training phase. It can be updated in a per frame basis, i.e., in a consecutive training phase.}
\begin{align}
p_{R}=\min\left\{p_{\max},\frac{w_{\text{th}}}{Q_{R}}\right\},
\label{ptr}
\end{align}   
where $Q_{R}\triangleq \mathbb{E}[\max_{i}\{\|\mathbf{g}_{i}\|^{2}\}^{m_{p}}_{i=1}]$, $w_{\text{th}}$ denotes the outage power threshold of the primary service with regards to the secondary transmission(s), which is assumed as a predetermined parameter, already known to all the secondary nodes, and $p_{\max}$ denotes the maximum achievable (unconstrained) power of the secondary system. 

\begin{cor}
The aforementioned transmission power at the receiver, $p_{R}$, is expressed as
\begin{align}
p_{R}=\left(\frac{1}{p_{\max}}+\frac{Q_{R}}{w_{\text{th}}}\right)^{-1},
\label{ptrclformmm}
\end{align}
where $Q_{R}$ is given in a closed form by
\begin{align}
Q_{R}=\sum^{m_{p}}_{i=1}\sum^{m_{p}}_{l=0}\frac{(-1)^{l}b^{N}_{R,i}}{l!\Gamma(N)}\underbrace{\sum^{m_{p}}_{n_{1}=1}\cdots \sum^{m_{p}}_{n_{l}=1}}_{n_{1}\neq \cdots \neq n_{l}\cdots \neq l}\sum^{N-1}_{k_{1}=0}\cdots \sum^{N-1}_{k_{l}=0}\left(\prod^{l}_{t=1}\frac{b_{R,k_{t}}}{k_{t}!}\right) \frac{\Gamma\left(N+\sum^{l}_{t=1}k_{t}+1\right)}{\left(b_{R,i}+\sum^{l}_{t=1}b_{R,n_{t}}\right)^{N+\sum^{l}_{t=1}k_{t}+1}},
\label{ptrclform}
\end{align}
where $b_{R,i}\triangleq (\beta_{i}-\hat{\beta}_{i})\alpha^{2}$ is a certain parameter corresponding to the link between the secondary receiver and the $i$th primary node ($1\leq i \leq m_{p}$).
\end{cor}

\begin{proof}
The proof is provided in Appendix \ref{appg}.
\end{proof}

Notice that $p_{R}$ is formed by using the channel estimates from the training phase. However, since the secondary receiver has full awareness of the channel time-variation (i.e., known $\alpha$), (\ref{ptrclform}) represents quite an efficient ceiling on the corresponding transmission power.\footnote{We assume that the secondary system is not aware of the instantaneous transmit/receive status for each primary node at each frame duration. Hence, $Q_{R}$ is formulated so as to protect \emph{all} the links between secondary receiver and primary nodes. In the simplified scenario when the secondary receiver knows the exact primary receiver at each frame (or when it is fixed), then $Q_{R}$ is still obtained from (\ref{ptrclform}) by setting $m_{p}=1$.}

The transmission power for all the secondary transmitters can be obtained quite similarly. In particular, referring back to the structure of $\mathbf{H}_{\text{tr}}=[\mathbf{h}_{1},\ldots,\mathbf{h}_{m_{p}},\mathbf{h}_{1},\ldots,\mathbf{h}_{m_{c}}]$ and $\mathbf{\Psi}=[\boldsymbol{\psi}_{1},\ldots,\boldsymbol{\psi}_{m_{p}},\boldsymbol{\psi}_{1},\ldots,\boldsymbol{\psi}_{m_{c}}]$ from (\ref{trainingsignal}), each secondary transmitter sends its pilot in its corresponding symbol-time duration. Notice that the pilots from primary nodes are foregoing the ones of the secondary nodes. Hence, each secondary transmitter can capture its channel response with regards to every primary node, by monitoring the first $m_{p}$ pilots, during the training phase. Then, using MMSE channel estimation (as explicitly described earlier), the $j$th transmission power at the corresponding secondary node, $p_{j}$, is determined by
\begin{align}
p_{j}=\left(\frac{1}{p_{\max}}+\frac{Q_{j}}{w_{\text{th}}}\right)^{-1}, \ \ 1\leq j\leq m_{c},
\label{pti}
\end{align}
where $Q_{j}$ is directly obtained from (\ref{ptrclform}), but denoting the $j$th secondary transmitter this time, instead of the secondary receiver.\footnote{We use channel reciprocity between primary and secondary nodes in order to formulate the aforementioned transmission powers in (\ref{ptr}) and (\ref{pti}).}

In the remaining symbol-time duration of training phase, where the secondary pilot symbol transmissions are sequentially established, $\{p_{j}\}^{m_{c}}_{j=1}$ are used to inform the secondary receiver about the corresponding channel states. 

\subsection{Unexpected Co-channel Interference at the Primary Nodes}
All the simultaneous secondary transmissions should not cause unexpected co-channel interference to any primary node. Thereby, the following condition should be satisfied
\begin{align}
\mathcal{I}_{j}\leq w_{\text{th}}, \ \ 1\leq j\leq m_{p},
\label{intcochannel}
\end{align} 
where $\mathcal{I}_{j}$ denotes the aggregate interfering power to the $j$th primary node from all the secondary transmitters.

In order to analytically evaluate $\mathcal{I}_{j}$, consider the case when the receiver senses the channel busy during an ongoing multi-node ($m_{c}$-fold) secondary transmission and then broadcasts its termination signal back to the secondary transmitters. Doing so, the worst case scenario in terms of unexpected co-channel interference includes the aggregate interfering power of $m_{c}+1$ signals. Assuming that the phases of the individual secondary signals fluctuate significantly, due to mutually independent modulation, the latter aggregate interference can be efficiently formed as an incoherent addition of the powers from $m_{c}+1$ signals \cite{j:yacoub}, which is a suitable model for practical applications \cite{j:incoherent}. Hence, for Rayleigh fading channels, each secondary signal power follows the exponential distribution and, thus, $\mathcal{I}_{j}$ is distributed by \cite[Eq. (5)]{j:rayleighinid}
\begin{align}
f_{\mathcal{I}_{j}}(x)=\sum^{m_{c}+1}_{i=1}\left(\displaystyle \prod^{i}_{\begin{subarray}{c}k=1\\k\neq i\end{subarray}}\textstyle \frac{p_{i}\bar{q}_{j,i}}{\left(p_{i}\bar{q}_{j,i}-p_{k}\bar{q}_{j,k}\right)}\right) \frac{\exp\left(-\frac{x}{p_{i}\bar{q}_{j,i}}\right)}{p_{i}\bar{q}_{j,i}},\ \ i=1,\ldots,m_{c},R
\label{ipdf}
\end{align}
where $\bar{q}_{j,i}\triangleq d^{-\omega_{i}}_{j,i}$ denotes the link distance between the $j$th primary and $i$th secondary node, while $R$ stands for the secondary receiver. Then, using the standard complementary CDF of exponential RVs, the probability of unexpected co-channel interference at the $j$th primary node is expressed as
\begin{align}
\text{Pr}[\mathcal{I}_{j}>w_{\text{th}}]=\int^{\infty}_{w_{\text{th}}}f_{\mathcal{I}_{j}}(x)dx=\sum^{m_{c}+1}_{i=1}\left(\displaystyle \prod^{i}_{\begin{subarray}{c}k=1\\k\neq i\end{subarray}}\textstyle \frac{p_{i}\bar{q}_{j,i}}{\left(p_{i}\bar{q}_{j,i}-p_{k}\bar{q}_{j,k}\right)}\right)\exp\left(-\frac{w_{\text{th}}}{p_{i}\bar{q}_{j,i}}\right).
\label{ccdfint}
\end{align}

\section{Numerical Results}
\label{Numerical Results}
\begin{figure}[!t]
\centering
\includegraphics[trim=1.5cm 0.3cm 2.5cm 1.2cm, clip=true,totalheight=0.25\textheight]{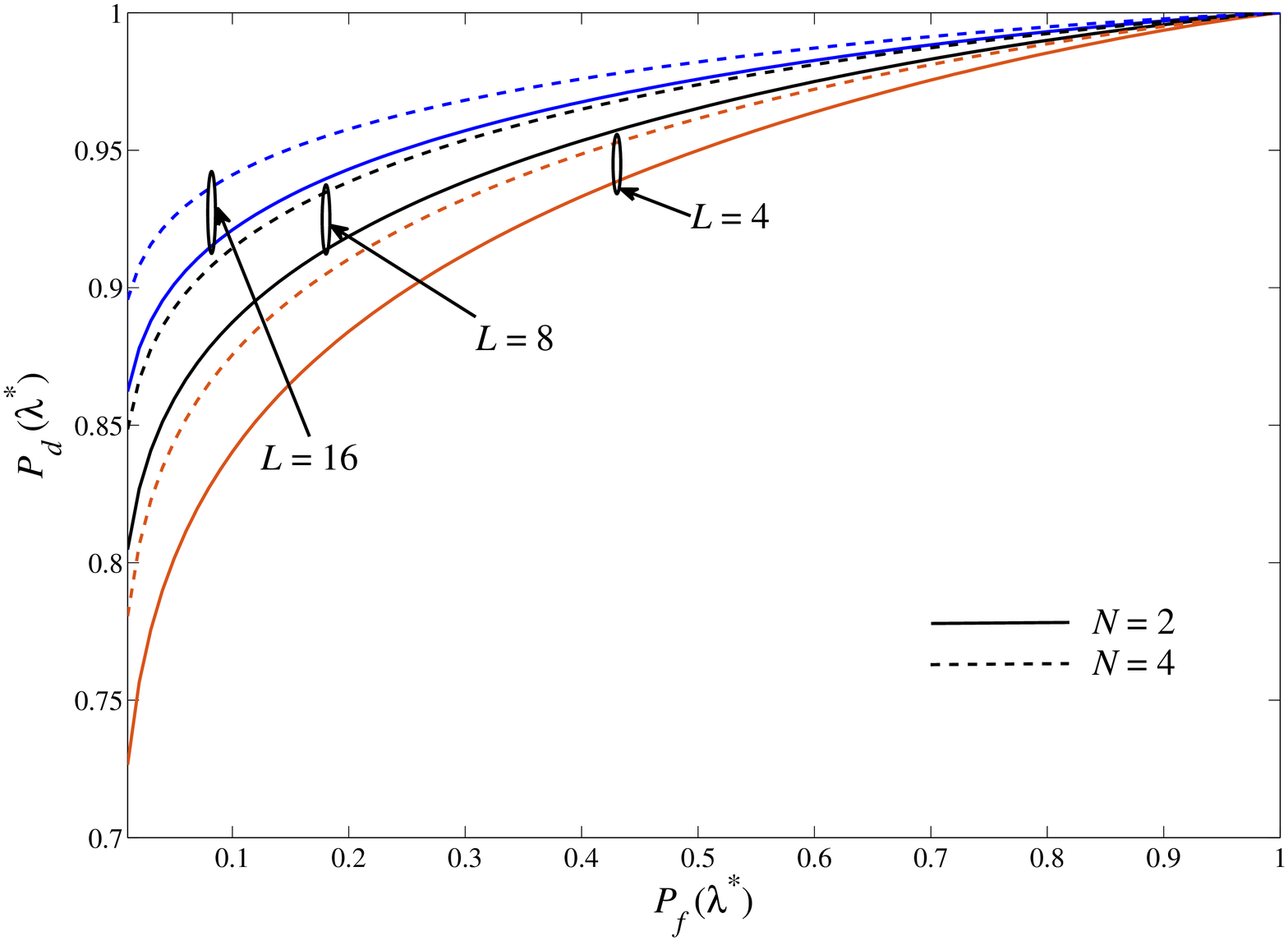}
\caption{Analytical ROC curve of the considered scheme for $m_{p}=4$ with $d_{1}=0.31$, $d_{2}=0.1$, $d_{3}=0.15$, and $d_{4}=0.2$.}
\label{fig2}
\end{figure}

\begin{figure}[!t]
\centering
\includegraphics[trim=1.3cm 0.3cm 2.5cm 1.2cm, clip=true,totalheight=0.25\textheight]{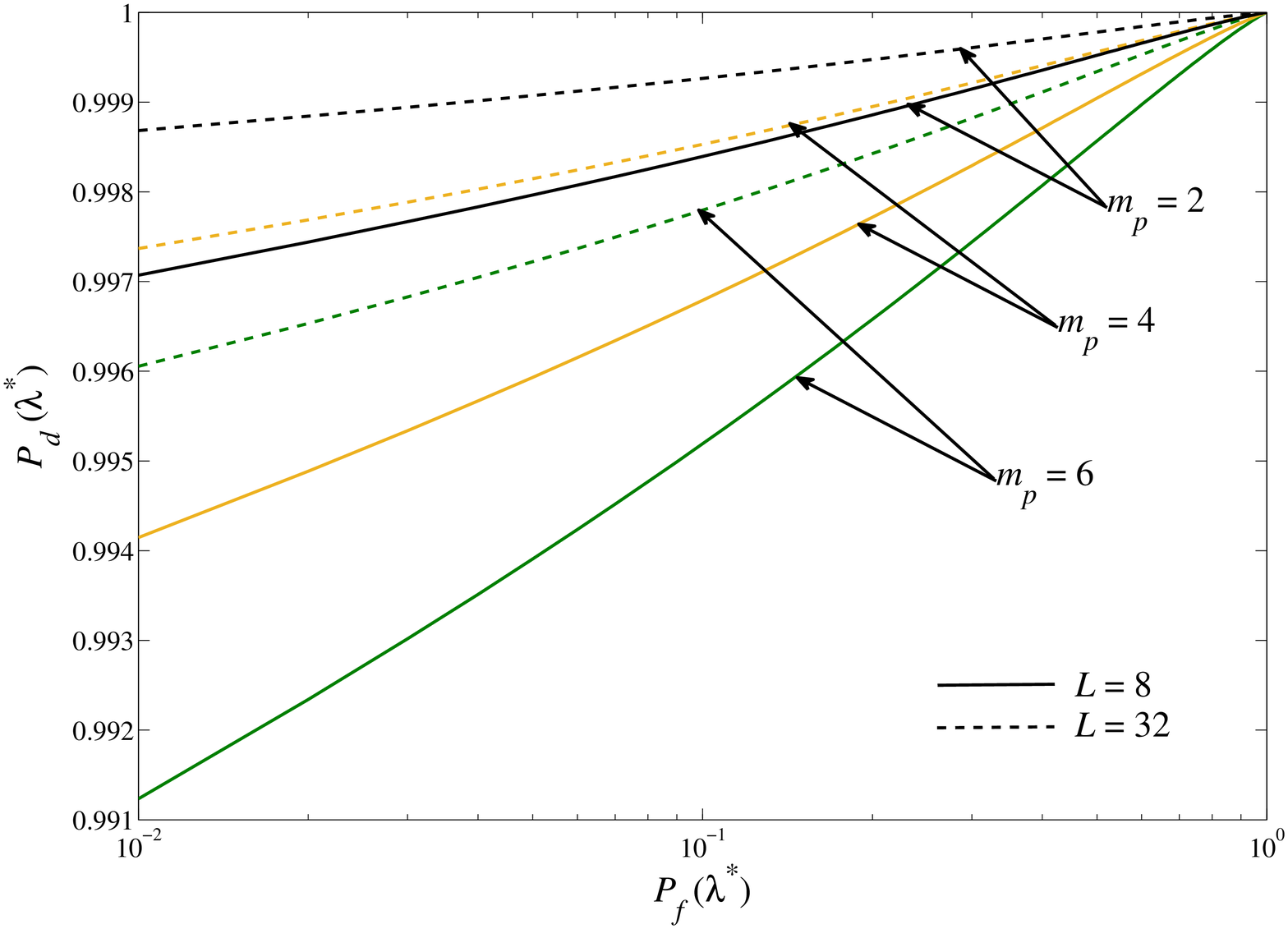}
\caption{Analytical ROC curve of the considered scheme for $N=2$, various numbers of primary nodes and identical link distances with respect to the secondary receiver, i.e., $\{d_{i}\}^{m_{p}}_{i=1}=0.1$.}
\label{fig3}
\end{figure}

\begin{figure}[!t]
\centering
\includegraphics[trim=1.5cm 0.3cm 2.5cm 1.2cm, clip=true,totalheight=0.25\textheight]{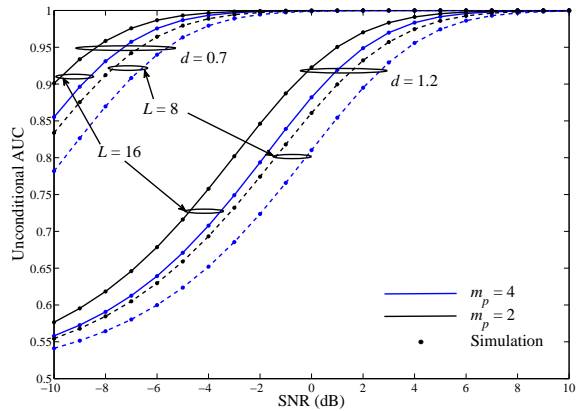}
\caption{Unconditional AUC vs. various SNR values of the primary nodes, considering identical link distances, while $N=4$.}
\label{fig4}
\end{figure}

\begin{figure}[!t]
\centering
\includegraphics[trim=1.5cm 0.3cm 2.5cm 1.2cm, clip=true,totalheight=0.25\textheight]{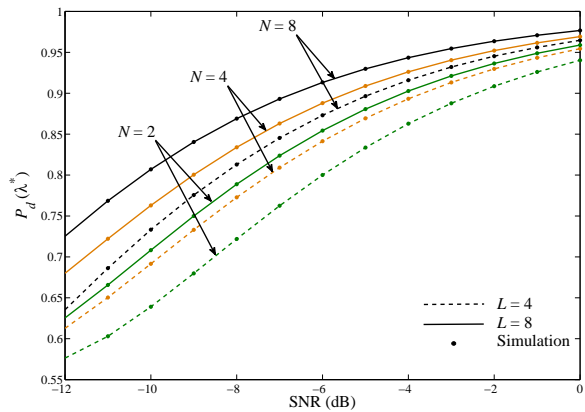}
\caption{Detection probability vs. various input SNR values for the primary nodes, when $\{d_{i}\}^{m_{p}}_{i=1}=0.3$, $m_{p}=2$ and $P_{f}=0.01$.}
\label{fig5}
\end{figure}

\begin{figure}[!t]
\centering
\includegraphics[trim=1.2cm 0.3cm 2.5cm 1.2cm, clip=true,totalheight=0.25\textheight]{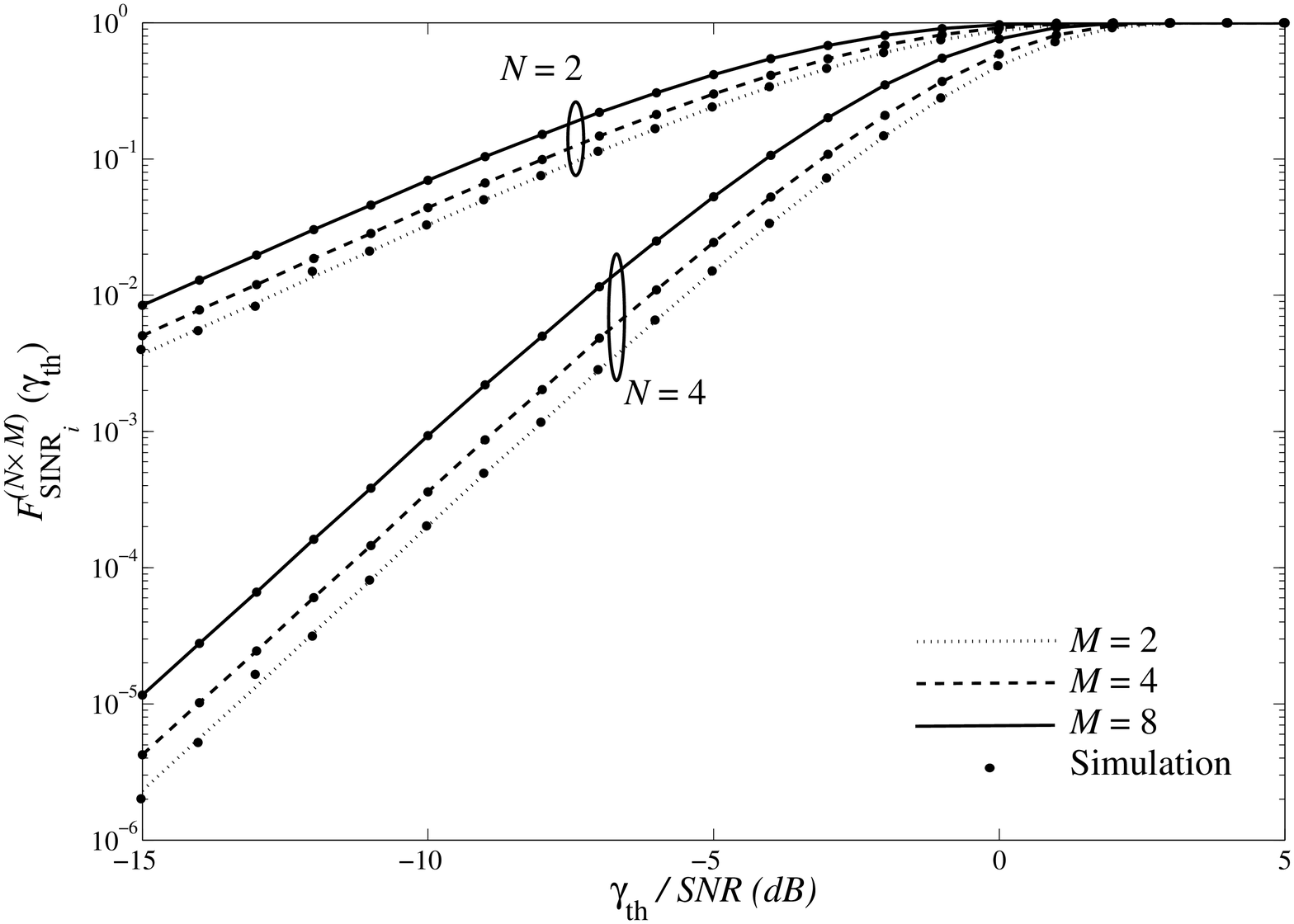}
\caption{CDF of the received SINR of a $N\times M$ system vs. various values of the normalized outage threshold, when $\{d_{i}\}^{M}_{i=1}=0.8+0.05i$.}
\label{fig6}
\end{figure}

\begin{figure}[!t]
\centering
\includegraphics[trim=1.5cm 0.3cm 2.5cm 1.2cm, clip=true,totalheight=0.25\textheight]{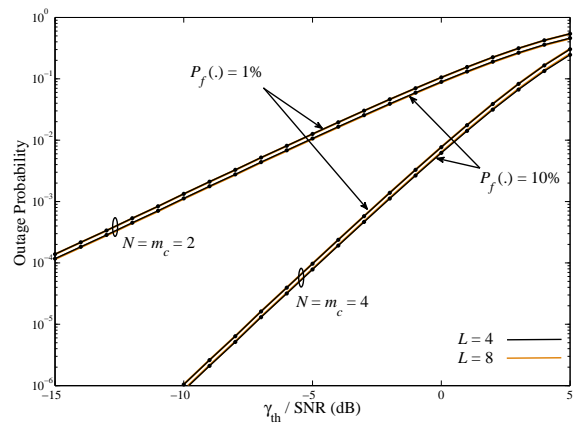}
\caption{Outage probability of the considered scheme vs. various values of the normalized outage threshold for identical link distances $\{d_{i}\}^{M}_{i=1}=0.8$, while $m_{p}=2$.}
\label{fig7}
\end{figure}

\begin{figure}[!t]
\centering
\includegraphics[trim=1.5cm 0.3cm 2.5cm 1.2cm, clip=true,totalheight=0.25\textheight]{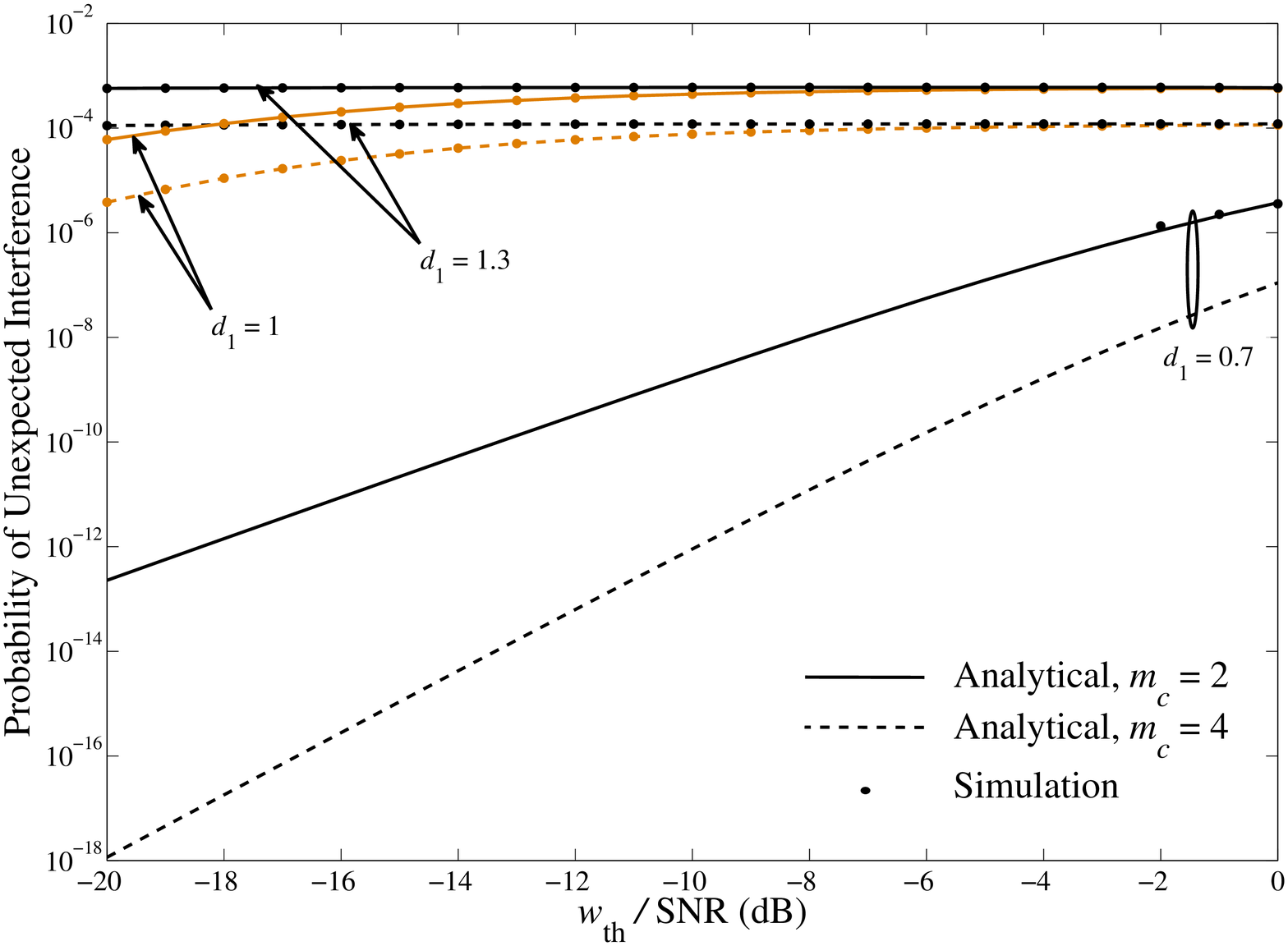}
\caption{Probability of unexpected interference to the primary nodes vs. various values of the normalized outage threshold (for the primary service), when $N=m_{p}=4$.}
\label{fig8}
\end{figure}

In this section, analytical results are presented and cross-compared with Monte-Carlo simulations. There is a good match between all the analytical and the respective simulation results and, hence, the accuracy of the proposed approach is verified. Henceforth, for notational simplicity and without loss of generality, we assume a common path-loss exponent $\omega=4$, corresponding to a classical macro-cell urban environment \cite[Table 2.2]{b:goldsmith}, while we fix the probability of transmission for all the primary nodes $P^{p}_{\mathcal{A}}=0.5$. Also, we set $\alpha=0.1$, $\sigma^{2}_{p}=1$ and $p_{\max}=20$dBm, while all the primary nodes use $p_{\max}$ for their transmissions. Some of the following numerical results are presented with respect to the input SNR of the primary nodes, referred as $\text{SNR}\triangleq p_{\max}/N_{0}$. 

Figures \ref{fig2} and \ref{fig3} present the ROC curves for the scenario of non-identical and identical statistics, respectively. Obviously, the performance of detection probability against false alarm probability improves for higher number of receive antennas. This is further enhanced when the available number of samples is increased. In addition, the presence of more primary users degrades the detection performance, since adding more unknown primary signals would be indistinguishable from noise. This result is in agreement with \cite[Fig. 7]{j:mckay}.

{\color{black}Similar conclusions can be drawn from Fig. \ref{fig4}, where the AUC performance is presented as a means of a more concrete performance tool in the entire energy threshold region, not only for the optimum $\lambda^{\star}$. In fact, the unconditional (average) AUC performance is depicted against different distances between the primary nodes and secondary receiver. It can be seen that the detection accuracy is reduced for far-distanced links, as expected, due to the unavoidable propagation attenuation on the received signals. Severe fading due to propagation losses results to noise-like signals. On the other hand, increasing SNR and/or the number of available samples for sensing result to a more accurate detection performance (i.e., AUC tends to unity).} Additionally, the presence of more receive antenna elements enhances the detection performance of the secondary receiver, as indicated in Fig. \ref{fig5}. This occurs due to the increased spatial diversity for higher $N$ values, which is manifested by capturing many different spatial observations for the same sample time-instance.

In Fig. \ref{fig6}, the CDF of the considered (virtual) $N\times M$ MIMO system is presented with non-identical statistics, where the analytical curves are based on (\ref{cdfsinr}). As can be seen, the performance improves for higher number of receive antennas with fixed number of simultaneously transmitting nodes, because of the emerged diversity gain. On the other hand, the performance is reduced for higher number of simultaneously transmitting nodes and a fixed number of receive antennas, since adding more co-channel interfering signal power degrades the total SINR. Importantly, the difference between the analytical curves and simulation points is rather marginal (there is a difference due to the approximation stage in (\ref{sinr}), yet it is rather negligible), which enhances the efficiency of the proposed scheme.

Moreover, Fig. \ref{fig7} demonstrates the total (unconditional) outage performance for some selected system scenarios. It is obvious that the target on the false alarm probability (i.e., the efficiency of detection scheme) and the available spatial DOFs play a key role to the outage probability. Finally, Fig. \ref{fig8} presents the probability of unexpected interference at the primary nodes. To preserve non-symmetrical distances (i.e., non-identical statistics appropriate for practical setups) let $\bar{q}_{j,i}\triangleq d_{1} (0.01i+0.01j)$. Interestingly, the latter probability is reduced for small link-distances between secondary and primary nodes and/or increased number of secondary transmitters. As an illustrative example, the cases when $d_{1}<0.7$ return negligible probability of unexpected interference for practical applications (i.e., below $10^{-6}$).

\section{Conclusion}
\label{Conclusion}
A D-MIMO cognitive (secondary) system was investigated, which operates under the presence of multiple primary nodes/users. A novel communication protocol was presented and evaluated when the secondary receiver utilizes MMSE detection. New analytical expressions for important performance metrics were derived in closed form, such as the outage and detection probabilities, {\color{black}the unconditional AUC}, and the impact of the transmission power from the secondary to the primary system. It was demonstrated that the probability of unexpected interference to the primary nodes remains quite low, by following the proposed guidelines, while the performance of the secondary system is directly associated with the signal detection accuracy.

\appendix
\subsection{Derivation of (\ref{phi})}
\label{appa}
\numberwithin{equation}{subsection}
\setcounter{equation}{0}
From (\ref{mse}), it holds that
\begin{align}
\nonumber
\text{MSE}_{i}&=\mathbb{E}\left[\left(s_{i}-\boldsymbol{\phi}_{i}^{\mathcal{H}}\mathbf{y}\right)\left(s_{i}-\boldsymbol{\phi}_{i}^{\mathcal{H}}\mathbf{y}\right)^{\mathcal{H}}\right]=1+\boldsymbol{\phi}^{\mathcal{H}}_{i}\mathbf{A}\boldsymbol{\phi}_{i}-s_{i}\mathbf{y}^{\mathcal{H}}\boldsymbol{\phi}_{i}-\boldsymbol{\phi}^{\mathcal{H}}_{i}\mathbf{y}s^{\mathcal{H}}_{i}\\
&=1+\left(\boldsymbol{\phi}^{\mathcal{H}}_{i}-(\mathbf{g}_{i}+\boldsymbol{\epsilon}_{i})^{\mathcal{H}}\mathbf{A}^{-1}\right)\mathbf{A} \left(\boldsymbol{\phi}^{\mathcal{H}}_{i}-(\mathbf{g}_{i}+\boldsymbol{\epsilon}_{i})^{\mathcal{H}}\mathbf{A}^{-1}\right)^{\mathcal{H}}-(\mathbf{g}_{i}+\boldsymbol{\epsilon}_{i})^{\mathcal{H}}\mathbf{A}^{-1}(\mathbf{g}_{i}+\boldsymbol{\epsilon}_{i}),
\label{msexpand}
\end{align} 
where $\mathbf{A}\triangleq \mathbb{E}[\mathbf{y}\mathbf{y}^{\mathcal{H}}]=\mathbf{C}\diag\{\beta_{j}\}^{M}_{j=1}\mathbf{C}^{\mathcal{H}}+N_{0}\mathbf{I}_{N}$ represents the covariance matrix of the received signal. Since only the first term of (\ref{msexpand}) depends on $\boldsymbol{\phi}_{i}$, the optimal solution that minimizes MSE$_{i}$ is $\boldsymbol{\phi}_{i}=\mathbf{A}^{-1}(\mathbf{g}_{i}+\boldsymbol{\epsilon}_{i})$. Finally, noticing that $(\mathbf{G}+\mathbf{E})=\mathbf{C}\diag\{\sqrt{\beta_{j}}\}^{M}_{j=1}$, (\ref{phi}) can be easily extracted.

\subsection{Derivation of (\ref{cdfsinr})}
\label{appb}
\numberwithin{equation}{subsection}
\setcounter{equation}{0}
From (\ref{sinrr}), it holds that $\text{Pr}[\text{SINR}_{i}\leq x]$ equals (\ref{cdfsinr}). Hence, the derivation of $F_{\Phi_{i}}(\cdot)$ is required to obtained the CDF of SINR for the $i$th transmitted stream. To this end, $F_{\Phi_{i}}(\cdot)$ is derived in a closed form as \cite[Eq. (11)]{j:GaoMMSE}
\begin{align}
F^{(N\times M)}_{\Phi_{i}}(y)=1-\exp\left(-\frac{N_{0}}{\beta_{i}}y\right)\sum^{N}_{i=1}\frac{\text{Ai}(y)\left(\frac{N_{0}}{\beta_{i}}y\right)^{i-1}}{(i-1)!},
\label{phicdfgao}
\end{align}
where $\text{Ai}(y)=1$ when $N\geq M+i-1$, or 
\begin{align*}
\text{Ai}(y)\triangleq \frac{1+\displaystyle \sum^{N-i}_{j=1}\:\:\sum_{1\leq n_{1}<\cdots<n_{j}\leq M}\scriptstyle \left(\frac{\beta_{n_{1}}}{\beta_{i}}\frac{\beta_{n_{2}}}{\beta_{i}}\cdots \frac{\beta_{n_{j}}}{\beta_{i}}\right)y^{j}}{\displaystyle \prod^{M}_{\begin{subarray}{c}n=1\\n\neq i\end{subarray}}\scriptstyle \left(1+\frac{\beta_{n}}{\beta_{i}}y\right)}
\end{align*}
when $N<M+i-1$. With the aid of \cite[Eqs. (65) and (70)]{j:mmsesuraweera}, a combined formation of the latter expression can directly be derived in (\ref{phicdf}) and (\ref{phicdfiid}) for the non-identical and identical statistics, correspondingly.

\subsection{Derivation of (\ref{out})}
\label{appe}
\numberwithin{equation}{subsection}
\setcounter{equation}{0}
Outage probability can be modeled by using the total probability theorem. Specifically, an outage event may occur if one of the following conditions hold: (a) when there is no active (transmitting) primary node, the receiver accurately senses the idle spectrum, and evaluates outage probability under the presence of $m_{c}$ independent signals; (b) when there is a miss detection event (i.e., the complement of detection probability) under the presence of one primary node averaging over its related probability; or (c) when there is a miss detection event under the presence of two primary nodes averaging over its related probability, and so on.

Condition (a) is explicitly defined in the first term of (\ref{out}), while conditions (b), (c) and so on are modeled by the second term of (\ref{out}) involving nested finite sum series (corresponding to the cases from $m_{c}+1$ to $m_{c}+m_{p}$ total active transmitted streams). Using (\ref{cdfsinr}), (\ref{pd}), (\ref{pf}), and (\ref{thr}) into (\ref{out}), outage probability can be directly computed in a closed-form, concluding the proof.

{\color{black}
\subsection{Derivation of (\ref{condauc})}
\label{appj}
\numberwithin{equation}{subsection}
\setcounter{equation}{0}
Plugging the first derivative of the false alarm probability $\partial P_{f}(\lambda')/\partial \lambda'=\lambda'^{2N L-1}\exp(-\lambda'^{2}/2)\\/(2^{N L-1}\Gamma(N L))$ and (\ref{pdcond}) into (\ref{pdpf}), we have that
\begin{align}
\text{AUC}(\mathcal{Y})=\frac{1}{2^{N L}\Gamma(N L)} \int^{\infty}_{0}t^{N L-1}\exp\left(-\frac{t}{2}\right)Q_{N L}\left(\sqrt{\frac{2 L \sigma^{2}_{p}\mathcal{Y}}{N_{0}}},\sqrt{t}\right)dt.
\label{pdpfff}
\end{align}
A closed-form solution for the latter expression was reported in \cite[Eq. (8)]{j:sofotasiosmarcum}. Thus, after some simple manipulations, (\ref{condauc}) is extracted.

\subsection{Derivation of (\ref{uncondauc})}
\label{appk}
\numberwithin{equation}{subsection}
\setcounter{equation}{0}
In principle, the unconditional AUC can be captured as $\overline{\text{AUC}}\triangleq \int^{\infty}_{0}\text{AUC}(x)f_{\mathcal{Y}}(x)dx$. Hence, using (\ref{condauc}) and (\ref{pdfmin}), while utilizing \cite[Eq. (7.621.4)]{tables}, we have that
\begin{align}
\nonumber
\overline{\text{AUC}}=&1-\sum^{N L-1}_{l=0}\sum^{m_{p}}_{s=1}\sum^{N-1}_{\begin{subarray}{c}t_{1}=0\\t_{1}\neq t_{s}\end{subarray}}\cdots \sum^{N-1}_{\begin{subarray}{c}t_{m_{p}}=0\\t_{m_{p}}\neq t_{s}\end{subarray}}\frac{(N L)_{l}\beta^{-t_{1}}_{1}\cdots \beta^{-N}_{s}\cdots \beta^{-t_{m_{p}}}_{m_{p}}\Gamma\left(\sum^{m_{p}}_{\begin{subarray}{c}l=1\\l\neq s\end{subarray}}t_{l}+N\right)}{l!2^{N L+l}t_{1}!\cdots t_{m_{p}}!\Gamma(N)\left(\sum^{m_{p}}_{t=1}\frac{1}{\beta_{t}}+\frac{L\sigma^{2}_{p}}{N_{0}}\right)^{\sum^{m_{p}}_{\begin{subarray}{c}l=1\\l\neq s\end{subarray}}t_{l}+N}}\\
&\times {}_2F_{1}\left(N L+l,\sum^{m_{p}}_{\begin{subarray}{c}l=1\\l\neq s\end{subarray}}t_{l},N L;\frac{L \sigma^{2}_{p}}{2N_{0}\left(\sum^{m_{p}}_{t=1}\frac{1}{\beta_{t}}+\frac{L\sigma^{2}_{p}}{N_{0}}\right)}\right).
\label{uncondauccc}
\end{align}
Using \cite[Eq. (07.23.03.0145.01)]{wolfram} into (\ref{uncondauccc}), we arrive at (\ref{uncondauc}).
}

\subsection{Derivation of (\ref{ptrclformmm}) and (\ref{ptrclform})}
\label{appg}
\numberwithin{equation}{subsection}
\setcounter{equation}{0}
Regarding the derivation of (\ref{ptrclformmm}) and recalling the Rayleigh fading condition, the PDF of the SNR for the probe message transmitted from the receiver becomes
\begin{align}
f_{X_{R}}(x)=
\left\{
\begin{array}{c l}     
    \frac{N_{0}\exp\left(-\frac{N_{0}x}{p_{\max}\bar{X}_{R}}\right)}{p_{\max}\bar{X}_{R}}, & Q_{R}<\frac{w_{\text{th}}}{p_{\max}},\\
    & \\
    \frac{N_{0}Q_{R}\exp\left(-\frac{N_{0}Q_{R}x}{w_{\text{th}}\bar{X}_{R}}\right)}{w_{\text{th}}\bar{X}_{R}}, & Q_{R}>\frac{w_{\text{th}}}{p_{\max}}.
\end{array}\right.
\label{fg2reporttt}
\end{align}
where $X_{R}$ and $\bar{X}_{R}$ denotes the instantaneous and average input SNR of the receiver. Hence, it yields that
\begin{align}
F_{X_{R}}(x)=1-\left(1-F_{X_{R}|p_{\max}}(x)\right)\left(1-F_{X_{R}|\frac{w_{\text{th}}}{Q_{R}}}(x)\right)=1-\exp\left(-\frac{N_{0}\left(\frac{1}{p_{\max}}+\frac{Q_{R}}{w_{\text{th}}}\right)x}{\bar{X}_{R}}\right).
\label{Fg2reporttt}
\end{align}
By differentiating (\ref{Fg2reporttt}), the corresponding PDF follows the classical exponential PDF with the yielded transmission power $p_{R}$ as defined in (\ref{ptrclformmm}).

Based on (\ref{eq1new}) and (\ref{remsignal}), we have that the actual channel matrix for the primary nodes can be expressed as $\mathbf{G}_{p}=\mathbf{C}_{p}\diag\{\sqrt{\beta_{i}}\}^{m_{p}}_{i=1}-\mathbf{E}_{p}$. Although the instantaneous values of $\mathbf{E}$ are not available, its distribution is known from (\ref{channeljoint}), using MMSE channel estimation. It easily follows that
\begin{align}
\mathbf{G}_{p}\overset{\text{d}}=\mathbf{C}_{p}\diag\left\{\sqrt{(\beta_{i}-\hat{\beta}_{i})\alpha^{2}}\right\}^{m_{p}}_{i=1}.
\label{gpdistr}
\end{align}
Thus, using the standard PDF/CDF expressions for chi-squared RVs, the maximum squared column norm of $\mathbf{G}_{p}$ is distributed as
\begin{align}
\nonumber
f_{\max_{i}\{\|\mathbf{g}_{i}\|^{2}\}^{m_{p}}_{i=1}}(x)&=\sum^{m_{p}}_{i=1}f_{b_{R,i}\mathcal{\chi}^{2}_{2 N}}(x)\prod^{m_{p}}_{\begin{subarray}{c}l=1\\l\neq i\end{subarray}}F_{b_{R,i}\mathcal{\chi}^{2}_{2 N}}(x)\\
&=\sum^{m_{p}}_{i=1}\frac{x^{N-1}\exp\left(-\frac{x}{b_{R,i}}\right)}{b^{N}_{R,i}\Gamma(N)}\prod^{m_{p}}_{\begin{subarray}{c}l=1\\l\neq i\end{subarray}}\left(1-\exp\left(-\frac{x}{b_{R,i}}\right)\sum^{N-1}_{k=0}\frac{\left(\frac{x}{b_{R,i}}\right)^{k}}{k!}\right).
\label{fmax}
\end{align}
By invoking the product expansion identities \cite[Eq. (6)]{j:miridakisrelay}, (\ref{fmax}) becomes after some simple manipulations
\begin{align}
\nonumber
f_{\max_{i}\{\|\mathbf{g}_{i}\|^{2}\}^{m_{p}}_{i=1}}(x)=&\sum^{m_{p}}_{i=1}\sum^{m_{p}}_{l=0}\frac{(-1)^{l}b^{N}_{R,i}}{l!\Gamma(N)}\underbrace{\sum^{m_{p}}_{n_{1}=1}\cdots \sum^{m_{p}}_{n_{l}=1}}_{n_{1}\neq \cdots \neq n_{l}\cdots \neq l}\sum^{N-1}_{k_{1}=0}\cdots \sum^{N-1}_{k_{l}=0}\left(\prod^{l}_{t=1}\frac{b_{R,k_{t}}}{k_{t}!}\right)\\
&\times \exp\left(-\left(b_{R,i}+\sum^{l}_{t=1}b_{R,n_{t}}\right)x\right) x^{\sum^{l}_{t=1}k_{t}+N-1}.
\label{fmaxx}
\end{align} 
Thereby, recognizing that $Q=\int^{\infty}_{0}x f_{\max_{i}\{\|\mathbf{g}_{i}\|^{2}\}^{m_{p}}_{i=1}}(x)dx$ and utilizing \cite[Eq. (3.381.4)]{tables}, (\ref{ptrclform}) is derived.

\ifCLASSOPTIONcaptionsoff
  \newpage
\fi

\bibliographystyle{IEEEtran}
\bibliography{IEEEabrv,References}

\end{document}